%% file: fermionic-mbz.tex
\documentclass[11pt]{article}
\usepackage{tikz,tikz-cd}
\usepackage{pgfplots}\pgfplotsset{compat=1.18}
\usepackage{mathrsfs}

\input{common_preamble}

\newcommand{\h}{\mathfrak h}

\title{Critical Fermions are Universal Embezzlers}

\author{Lauritz van Luijk, Alexander Stottmeister, Henrik Wilming}
\date{\normalsize Institut f\"ur Theoretische Physik, Leibniz Universit\"at Hannover, \\ Appelstraße 5, 30167 Hannover, Germany\\[2ex]\today}

\begin{document}

\maketitle

\begin{abstract}
Universal embezzlers are bipartite quantum systems from which any entangled state may be extracted to arbitrary precision using local operations while perturbing the system arbitrarily little. 
We show that universal embezzlers are ubiquitous in many-body physics: The ground state sector of every local, translation-invariant, and critical free-fermionic many-body system on a one-dimensional lattice is a universal embezzler if bi-partitioned into two half-chains.
The same property holds for the locally interacting spin chains that are dual via the Jordan-Wigner transformation.
Universal embezzlement already manifests at finite system sizes, not only in the thermodynamic limit: For any finite error and any targeted entangled state, a finite length of the chain is sufficient to embezzle said state within the given error.
On a technical level, our main result establishes that the half-chain observable algebras associated with ground state sectors of the given models are type III$_1$ factors.
\end{abstract}

\tableofcontents

\section{Introduction}
The entanglement structure of ground states of many-body systems plays a crucial role in our understanding of quantum phases of matter. It can be used to detect quantum phase transitions \cite{osborne_entanglement_2002,vidal_entanglement_2003}, and much work has been invested in obtaining a deep understanding of the entanglement structure of gapped ground states dominated by the area law \cite{eisert_colloquium_2010}. An important outcome of this work has been the introduction of tensor networks as both numerical and theoretical tools to understand such systems \cite{schollwock_density-matrix_2005,schollwoeck_density-matrix_2011,cirac_matrix_2021}.  Gapped phases of matter exhibit rich entanglement properties, for example, encoded in their topological orders with the associated anyonic excitations and quantum error-correcting codes \cite{wen_topological_1990,kitaev_fault-tolerant_2003,kitaev_anyons_2006,wen_topological_2013}.

The entanglement of ground states at quantum phase transitions shows a very distinct behavior from the gapped phases. Critical many-body systems are expected to have a scaling limit described by a conformal field theory (CFT) \cite{di_francesco_conformal_1997} (cf.~\cite{osborne2023cft_lattice_fermions,osborne2023ising} for recent rigorous results in free-fermion systems). On large scales, one can hence expect similar behavior as in conformal field theories, for which various local quantifiers of entanglement, such as the geometric entropy of single or multiple intervals, have been computed \cite{srednicki_entropy_1993,holzhey_geometric_1994, calabrese2004entanglement_entropy_qft, calabrese2009entanglement_entropy_cft, casini2009entanglement_entropy_qft, headrick2010entanglement_renyi_entropies, longo2018relative_entropy_cft}.
This has been confirmed in specific models; see, for example, \cite{vidal_entanglement_2003,latorre_ground_2004}. 
In one spatial dimension, a hallmark feature of such critical ground states is that the entanglement entropy of a connected region grows logarithmically with the size of the region instead of the constant upper bound implied by the area law. 
Therefore, if we partition an infinite chain into two half-chains, they are necessarily infinitely entangled.

In this work, we explain how the infinite entanglement in critical free-fermionic many-body systems can be characterized using the information-theoretic task of \emph{embezzlement of entanglement} introduced in \cite{van_dam2003universal}. The recent result that embezzlement provides an operational interpretation for different ways of being infinitely entangled \cite{short_paper,long_paper} prompts an application to many-body physics.
In the task of embezzlement, two agents (say Alice and Bob) can each access one half of a shared entangled resource as well as a local subsystem. They are asked to produce, via local operations and without communication, an entangled state on their two local subsystems while perturbing the resource state as little as possible.
The resource is called an \emph{embezzling state} if this is possible for \emph{every entangled state} (of arbitrary dimension) \emph{to arbitrary precision} while perturbing the resource state \emph{arbitrarily little} \cite{short_paper,long_paper}.
More precisely, a pure state $\ket\Omega_{AB}$ on a Hilbert space $\H$ modeling a bipartite system is embezzling if for every $\eps>0$ and any finite-dimensional entangled state $\ket\Psi_{AB} \in \mathbb C^d\otimes \mathbb C^d$ (for any $d$) there exist unitaries $u_{AA'}$ and $u_{BB'}$ of Alice and Bob, respectively, such that
\begin{align}\label{eq:mbz}
    u_{AA'} u_{BB'} \big(\ket\Omega_{AB}\otimes \ket 0_{A'}\ket 0_{B'}\big) \approx_\eps \ket\Omega_{AB} \otimes\ket\Psi_{A'B'}, 
\end{align}
where $\ket 0_{A'}\ket 0_{B'}\in \CC^d\otimes\CC^d$ is a product state. 

A bipartite physical system on which every pure state is an embezzling state is called a \emph{universal embezzler}. 
It is far from obvious that embezzling states, much less universal embezzlers, exist. 
Indeed, such systems clearly require infinite amounts of entanglement. Ref.~\cite{long_paper} showed that universal embezzlers can be constructed mathematically and are deeply linked to the classification of von Neumann algebras. The results imply that, unsurprisingly, a universal embezzler must have infinitely many degrees of freedom. On a technical level, a bipartite quantum system is a universal embezzler if and only if the local observable algebras of the two parts of the system are so-called type $\III_1$ von Neumann algebras.

One may wonder whether universal embezzlers are purely mathematical objects or actually appear in physics. 
In \cite{long_paper}, it was shown that relativistic quantum field theories provide examples of universal embezzlers. While this transparently explains the maximum Bell inequality violations of the vacuum \cite{summers1985vacuum}, it currently appears to be necessary that both agents must control exceedingly large portions of their respective halves of spacetime to achieve acceptable errors for non-trivial target states, which renders this result rather impractical.

In this work, we show that critical, fermionic many-particle systems provide an abundance of examples of universal embezzlers: On a one-dimensional lattice, any free fermionic, translation-invariant, local model with a nontrivial Fermi surface is a universal embezzler in the ground state sector.

\begin{figure}
    \centering    \includegraphics{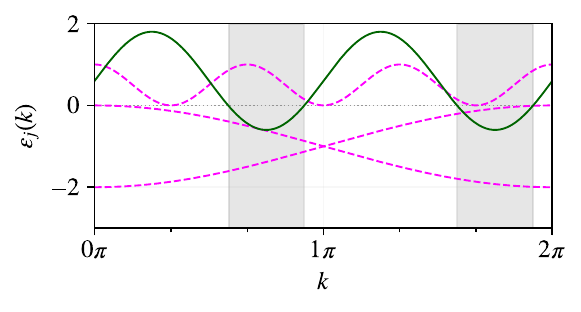}
    \caption{We consider translation-invariant free-fermionic models on a one-dimensional lattice, where the energy bands $\varepsilon_j(k)$ are piecewise continuous functions of the wave number $k$ and at least one of the energy bands is \emph{critical}, i.e., exhibits a non-trivial Fermi surface, here illustrated in solid green. The energy bands depicted in dashed magenta do not count as critical in our sense because the subsets of wave numbers $k$ for which they are not strictly positive or not strictly negative, respectively, consist of only finitely many points: The bands do not cross the Fermi energy at $0$. }
    \label{fig:dispersion}
\end{figure}

Moreover, we show that the large-scale structure of entanglement required for universal embezzlers has strong implications for the finite-size ground states: For any specified finite error $\eps >0$ and any fixed dimension $d$ there is a sufficiently large system size $N(\eps,d)$ such that all entangled states of local dimension $d$ may be embezzled up to error $\eps$ from the ground state of the system.  The latter result is true beyond free-fermionic systems and, in fact, holds for \emph{all} universally embezzling many-body systems with a unique ground state. 

Thus, the relationship between embezzling states and embezzling families is roughly similar to that between the sharp (and idealized) notion of phase transitions in infinite systems and their approximate manifestation in large but finite systems.
This shows that the study of the entanglement in the thermodynamic limit is also relevant for finite systems.  

To obtain our results, we have to overcome several obstacles. First, we have to argue that the parity super-selection rule for fermions does not prohibit the use of the results of \cite{long_paper} on embezzlement. This requires new operator algebraic results on the structure of fermionic quantum theory. In particular, we need to establish a property called \emph{Haag duality} for the physically relevant even operators. 
Second, we need to show that the local von Neumann algebras in the ground state sectors are of type $\III_1$ for the full class of models that we consider. This result generalizes a prior result of Matsui, who showed the type $\III_1$ property for the XX spin chain \cite{matsui_split_2001,keyl2006}. To do so, we combine general operator algebraic results on quasi-free fermionic states by Powers and St\o rmer \cite{powers1970free_states} with the theory of Toeplitz operators.
Toeplitz theory has previously been used to obtain precise calculations of the asymptotic behavior of local entanglement quantifiers in spin chains and free fermionic models \cite{jin_quantum_2004,eisert_single-copy_2005,its_entanglement_2005,keating_random_2004,keating_entanglement_2005} and is closely connected to the classification of one-dimensional topological phases of quantum walks \cite{cedzich_topological_2018,cedzich_complete_2018}. 
Third, we have to relate the entanglement structure in the thermodynamic limit to that in finite systems.
In the following, we discuss these points in more detail with a focus on providing intuition and a high-level understanding. The precise mathematical arguments are given in sections \ref{sec:bipartite} to \ref{sec:fermin_mbz_univ}.

\section{Results}\label{sec:results}
We consider fermionic many-body Hamiltonians on the infinite one-dimensional lattice with $b$ orbitals per lattice site. The algebra of operators is generated by the creation and annihilation operators $a(\xi),a^\dagger(\eta)$ with $\ket\xi,\ket\eta\in \h := \ell^2(\ZZ)\otimes \mathbb C^b$ obeying the canonical anti-commutation relations. We use the standard bases $\ket x \in \ell^2(\ZZ)$ and $\ket{j}\in \CC^b$ to define $a_j(x):= a(\ket x\otimes \ket j)$.  The class of Hamiltonians that we consider are second-quantized single-particle Hamiltonians $h$ on $\h$ of the form
\begin{align}\label{eq:hamiltonians}
    H = \sum_{x,y\in\ZZ} \sum_{l,m=1}^b h_{lm}(x-y) a_{l}^\dagger(x) a_m(y),\quad h_{lm}(x-y) = \bra{x}\otimes \bra{l} h \ket{y}\otimes\ket{m},
\end{align}
where we already made use of translation invariance. We assume that the Fourier transforms $\hat h_{lm}(k)$ for $k\in S^1 \cong [0,2\pi)$ exist. Then the hermitian matrix $\hat h(k)$ can be diagonalized for every $k\in S^1$. We denote the projection onto the strictly positive eigenvalues of $\hat h(k)$ by $\hat p_+(k)$ and require that $h$ is sufficiently local in space so that $\hat p_+$ is a piecewise continuous function with at most finitely many discontinuities.\footnote{{Technically, $\hat p_+$ is an element of the space $L^\infty(S^1,M_b(\mathbb{C}))$. Thus, it is only defined up to the values on sets of measure zero. In particular, we can assume it is continuous from the left.}} 
This is guaranteed if the Hamiltonian is short-ranged but also allows for long-range couplings.
We say that the Hamiltonian $H$ is \emph{critical} if $\hat p_+$ has at least one discontinuity.
In the case of a short-ranged Hamiltonian, we can think of energy bands described by eigenvalue functions $\epsilon_j$ that are analytic on $S^1$ up to (at most) a single point of discontinuity and the same is true for the associated spectral projections \cite{kato_perturbation_1995}.
It is then sufficient that at least one of the energy bands $\epsilon_j$ is negative on a nontrivial interval $I \subset S^1$ and positive on a nontrivial interval in $S^1\setminus I$, see figure~\ref{fig:dispersion} for an illustration. We emphasize that being critical in our sense is a stronger condition than merely being gapless: If an energy band has isolated points at which it vanishes, it is not critical, but the Hamiltonian does not have a gap above the ground state. The simple hopping Hamiltonian gives a paradigmatic example for the given class of Hamiltonians
\begin{align}\label{eq:xx}
    H_{\text{XX}} = \sum_{x} \big(a^\dagger(x)a(x+1) + a^\dagger(x+1) a(x)\big). 
\end{align}
A different example is the Su-Schrieffer-Heeger model \cite{su_soliton_1980} at the topological phase-transition point, given by
\begin{align}
    H_{\text{SSH}} = \sum_{x} \left(a_1^\dagger(x)a_2(x) + a_2(x)^\dagger a_1(x) + a_1^\dagger(x)a_2(x+1) + a_2(x+1)^\dagger a_1(x)\right),
\end{align}
which features a topologically non-trivial band structure (see section \ref{sec:fermin_mbz_univ}).
The topology of the band structure is insignificant for our results.

In systems with infinitely many degrees of freedom, states are most usefully defined as expectation value functionals on the algebra of local observables describing the system.
If $\h_0$ is the kernel of $h$, and $p_+$ is the spectral projection onto the strictly positive part of $h$, the pure ground states of $H$ are parametrized by 
projections $p_0$ onto subspaces of $\h_0$. They are given by the quasi-free (fermionic Gaussian) states $\omega_{p}$ on the algebra of canonical anti-commutation relations $\CAR(\h)$ on $\h$ determined by
\begin{align}
    \omega_{p}\big(a_l(x) a_m(y)^\dagger\big) = \bra x\otimes \bra l p\ket y\otimes\ket m ,\quad p = p_+ + p_0,
\end{align}
and Wick's theorem for higher-order correlators. In the following, we set $p_0=0$, but our results also hold for a wide range of choices of $p_0$ (if $\h_0$ is non-trivial), see section \ref{sec:fermin_mbz_univ}.

A bipartition of the system is defined by a projection $q$ onto a subspace $q\h$: The observables generated by $a(\xi),a^\dagger(\eta)$ with $\ket\xi,\ket\eta\in q\h$ belong to Alice, and those generated by $\ket\xi,\ket\eta\in q^\perp \h=(1-q)\h$ belong to Bob. 
We consider a bipartition of the spatial lattice into two half-infinite chains so that $q$ is the projection onto $\ell^2(\ZZ_+)\otimes \CC^b$.
The quantum state associated with Alice is simply the restriction of $\omega_p$  to the corresponding operators, which we can identify with $\omega_{qpq}$ (see section \ref{sec:fermin_mbz_univ}).

Physical observables of a fermionic system are necessarily even operators, i.e., invariant under parity transformations, as dictated by the parity super-selection rule \cite{wick_intrinsic_1952}.
We write $\CAR_e(\h)$ for the algebra of even operators generated by the creation and annihilation operators with vectors from $\h$. The even operators $\CAR_e(\h)$ operate on the even part $\mc H_e$ of the full Fock-space $\mc H$ associated with $\omega_p$, which is represented by a vector $\ket\Omega\in \mc H_e$. The local observable algebras $\CAR_e(q\h)$ and $\CAR_e(q^\perp\h)$ act on $\mc H_e$. Their weak closures are the \emph{von Neumann algebras} $\M_A$ and $\M_B$, belonging to Alice and Bob, respectively. The weak topology is the one making expectation value functionals with respect to vectors on $\H$ continuous, as physically expected. The triple $\left (\M_A,\M_B,\H_e\right)$ defines a bipartite system in the sense of \cite{long_paper}. 

To use the results of \cite{long_paper} to establish the universal embezzlement property, we have to ensure that i) $\M_A$ and $\M_B$ together generate all bounded operators on $\mc H_e$ and that ii) the two operator algebras exhaust all quantum degrees of freedom, which is captured by a condition called \emph{Haag duality}: It says that the commutant $\M_A':= \{ x \in \mathcal B(\mathcal H_e) : [x,y] = 0\ \forall y\in \M_A\}$ is precisely given by $\M_B$. 
If $\h$ is infinite-dimensional, the condition i) does a priori not automatically ensure condition ii). 
Settling this issue is our first result: 
\begin{theorem}
    If $\omega_p$ is a quasi-free pure state on $\CAR_e(\h)$, and $q$ a projection on $\h$, then $\M_A' = \M_B$.
\end{theorem}

Von Neumann algebras can be classified into different types. 
The main result of \cite{long_paper} states that the bipartite system we consider is a universal embezzler if and only if $\M_A$ is of so-called type $\III_1$ (which automatically ensures that the same holds for $\M_B = \M_A'$). 
The type of $\M_A$ is determined by the properties of Alice's reduced state $\omega_{qpq}$, which in turn is completely determined by the operator $qpq$ on $\h$. In fact, it was shown by Powers and St\o rmer in \cite{powers1970free_states} that the type of $\M_A$ is entirely determined by the spectrum of the operator $qpq$. To get an intuition for this, assume that $\h$ is finite-dimensional. If $qpq$ has eigenvectors $\ket{\xi_j} \in q\h$ with eigenvalues $\lambda_j$, we find from Wick's theorem that
\begin{align}
\omega_{qpq}\left(a^\dagger(\xi_{j_1}) \cdots a^\dagger(\xi_{j_m})a(\xi_{j_m})\cdots a(\xi_{j_1})\right) = \prod_{k=1}^m \omega_{qpq}\left(a^\dagger(\xi_{j_k})a(\xi_{j_k})\right) = \prod_{k=1}^m (1-\lambda_{j_k}). 
\end{align}
Each normal mode $\ket{\xi_j}$ corresponds to a single fermionic mode, and we just saw that all of them are uncorrelated. We can think of Alice's state $\omega_{qpq}$ as describing a tensor-product of a finite number of uncorrelated two-level systems, each represented by a density matrix $\rho_j = \mathrm{diag}(\lambda_j,1-\lambda_j)$.  The $\lambda_j$ measure how mixed the local state of Alice is. 
In the bipartite pure state, each normal mode of Alice is entangled with precisely one normal mode of Bob and carries at most one ebit of entanglement, which happens if $\lambda_j=1/2$. This entanglement structure is the characteristic property of Gaussian fermionic states.

In our case, $\h$ is infinite-dimensional. The type of the local observable algebra of an infinite spin chain described by an infinite tensor product $\otimes_j \rho_j$ has been studied by Araki and Woods \cite{araki1968factors}.  The resulting algebra has type $\III_1$ if and only if the ratios of products of eigenvalues of the $\rho_j$ are dense in $\RR_+$. Intuitively, this means that ``all amounts of entanglement occur with arbitrary multiplicity''. The results have been transferred to quasi-free fermionic states by Powers and St\o rmer, where the spectrum of $qpq$ can be continuous.

To determine the spectrum of $qpq$ we make use of the theory of block Toeplitz operators. This allows us to show:
\begin{theorem}\label{thm:intro-fermions}
If $H$ is critical, there exists a gauge-invariant quasi-free, pure, and translation-invariant ground state such that $\M_A$ is a type $\III_1$ factor. 
\end{theorem}
In the case of a linear, gapless dispersion relation, the scaling limit of $H$ is expected to be well described by a conformal field theory (CFT) \cite{belavin1984infinite_conformal1, di_francesco_conformal_1997, fernandez1992random_walks}. It is known that the local observable algebras in CFTs are of type $\III_{1}$ \cite{brunetti1993modular_structure, gabbiani1993oa_cft}. Hence, their vacua yield embezzling states. We, therefore, expect that any local Hamiltonian with a CFT as its scaling limit has a universally embezzling ground state sector. 

\cref{thm:main-fermions} also applies to spin systems that are dual to the corresponding fermionic Hamiltonians by the Jordan-Wigner transformation \cite{jordan_uber_1928,lieb_two_1961,evans1998qsym}. This needs to be clarified since the relevant local observable algebras $\M_A$ are different due to the missing parity super-selection rule. Moreover, we show that Haag duality holds. This implies that these spin systems are also universal embezzlers. 

Our final result shows that the embezzling property of ground states in the thermodynamic limit also descends to finite systems. We call a family of pure states $\ket{\Omega_n}$ an \emph{embezzling family} if for every dimension $d$ and every error $\eps>0$ there exists a number $n(\eps,d)$ such that every state $\ket\Psi\in \CC^d\otimes \CC^d$ can be embezzled from $\ket{\Omega_m}$ with $m\geq n(\eps,d)$ up to error $\eps$ in the sense of \eqref{eq:mbz} \cite{leung_characteristics_2014,van_dam2003universal}.

\begin{theorem}[informal]\label{thm:main-families} Consider a family of local Hamiltonians $H_n$ on increasingly long chains with a unique ground state $\ket \Omega$ in the thermodynamic limit. If $\ket \Omega$ is an embezzling state, then there exists a sequence of ground states $\ket{\Omega_n}$ of $H_n$ that is an embezzling family. 
\end{theorem}
More generally, this result follows from a theorem showing that any embezzling state on a \emph{hyperfinite} von Neumann algebra gives rise to an embezzling family.
Thus, if we only want to embezzle states of a fixed dimension and up to a fixed error, we do not actually need the full thermodynamic limit. The thermodynamic limit provides a sharp distinction but is not physically required. We note that assuming the uniqueness of the ground state $\ket{\Omega}$ is not strictly necessary (cf.~\cref{rem:uniqueness}).

\section{Discussion}
We have shown that a very large class of effective models for critical, fermionic many-body systems have a remarkable large-scale structure of entanglement by acting as universal embezzlers when interpreted as bipartite systems. 
The embezzlement property crucially relies on the choice of bipartition. For example, suppose instead of considering two half-chains, Alice had access to the odd sublattice, and Bob had access to the even sublattice. In that case, it follows from the results of \cite{wilming_quantized_2023} and \cite{powers1970free_states} that the algebras $\M_A$ have type $\II_1$ instead of $\III_1$ for any sublattice symmetric Hamiltonian, such as \eqref{eq:xx}. According to \cite{long_paper}, no type $\II$ von Neumann algebra can host embezzling states. 

It has been shown that gapped one-dimensional, local many-body systems cannot be universal embezzlers as a consequence of the area law, which implies that $\M_A$ has type $\I$ \cite{matsui_boundedness_2013}. 
This raises the natural question of whether \emph{every} gapless, translation-invariant, local many-body Hamiltonian has a ground state that leads to a universal embezzler in the thermodynamic limit. It is tempting to conjecture that this is always the case if the system is critical in the sense of a CFT scaling limit, meaning that the large-scale observables of a half-space form a $\III_{1}$ factor. Interestingly, the Motzkin spin chain \cite{movassagh_supercritical_2016} provides an example of a system where the gap closes too quickly for a CFT scaling limit to exist. It would be of great interest to determine the type of the associated half-chain von Neumann algebra. Is it also of type $\III_1$? Does Haag duality hold?

More generally, in conjunction with \cite{long_paper}, our results motivate the study of the large-scale structure of entanglement in interacting critical many-body systems.
For example, it would be interesting to know whether the type of the half-chain algebras can be deduced from scaling laws of entanglement entropies. 

\paragraph {\bf Acknowledgements.}
We would like to thank Reinhard F.~Werner, Tobias J.~Osborne, and Robert Fulsche for helpful discussions. We thank Roberto Longo and Jan Derezi{\'n}ski for sharing their knowledge about index theory and quasi-free systems. LvL and AS have been funded by the MWK Lower Saxony via the Stay Inspired Program (Grant ID: 15-76251-2-Stay-9/22-16583/2022).

\section{Abstract bipartite systems and universal embezzlement}
\label{sec:bipartite}

In the remainder, we assume familiarity with basic operator algebraic notions. 
We refer to \cite{takesaki1} for a complete reference and to \cite[Sec.~3.1]{long_paper} for a brief overview.
We note a few notional conventions:

\paragraph{Notation.}
We use bra-ket notation for inner products. 
Given a C*-algebra $\mathcal A$ acting on a Hilbert space, $[\mathcal A\Omega]$ denotes the closed linear subspace generated by acting with $\mc A$ on the vector $\Omega$. We sometimes identify a subspace with the orthogonal projection onto the subspace to simplify notation.
If $P$ is a projection operator on a Hilbert space, we write $P^\perp := 1-P$. The spectrum of a linear operator $A$ is denoted as $\spec(A)$ and its essential spectrum as $\essspec(A)$. The algebra of complex $n\times n$ matrices is denoted as $M_n(\CC)$. Throughout, we use $A^*$ and not $A^\dagger$ to denote the adjoint of a linear operator, but we keep the notation $a(\xi),a^\dagger(\eta)$ for annihilation and creation operators (see below). Given $\ket\xi\in\h = \ell^2(\ZZ)\ox \CC^b$, we write 
\begin{align}
    \ket{\xi(x)} = (\bra{x}\ox \1)\ket{\xi},\quad \ket{\hat \xi(k)} := \sum_{x\in\ZZ}\e^{-ikx} \ket{\xi(x)},\quad x\in\ZZ,\ k\in S^1\cong [0,2\pi),
\end{align}
for the position and momentum representation via the unitary Fourier-transform $\h\rightarrow \h$. Note that $\ket{\xi(x)}, \ket{\hat\xi(k)} \in \CC^b$. \\

All our results rest upon the concept of a bipartite system in the sense of \cite{long_paper}, see also \cite{keyl_infinitely_2003,van_luijk_schmidt_2023}.
We denote by $\H$ the (separable) Hilbert space describing the joint system. 
The observable algebras of Alice and Bob are commuting von Neumann algebras $\M_A$ and $\M_B$ on $\H$, which we assume to satisfy \emph{Haag duality} \cite{keyl2006}. The latter property states that every symmetry of Alice's part, i.e., every bounded operator $b\in \B(\H)$ that commutes with all operators $a\in \M_A$, is contained in the algebra of Bob's part $\M_B$:
\begin{align}\label{eq:haag}
    \M_B  = \M_A' := \{ b\in \B(\H)\ :\ [a,b]=0\ \, \forall a\in \M_A\},
\end{align}
where the prime denotes taking the commutant.
Thus, a bipartite system $(\M_A,\M_B,\H)$ is fully specified by the von Neumann algebra $\M_A$ acting on $\H$.
Conversely, every von Neumann algebra $\M$ acting on a Hilbert space $\H$ determines a bipartite system via $\M_A=\M$, $\M_B=\M'$.

In many circumstances, it is reasonable to assume, or one can explicitly show, that the systems of Alice and Bob contain no classical degrees of freedom in the sense that the observable algebras have trivial center:
\begin{equation}\label{eq:center}
    Z(\M_A) = Z(\M_B)= \CC1,
\end{equation} 
where the center $Z(\N)$ of a von Neumann algebra $\N$ is $Z(\N):= \N\cap\N'$.
A von Neumann algebra with a trivial center is called a factor. 

Following \cite{long_paper,short_paper} a bipartite system is a \emph{universal embezzler} if for every vector $\ket\Omega_{AB}\in \H$, every finite-dimensional pure entangled state $\ket\Psi_{A'B'}\in \CC^d\otimes \CC^d$, and every error $\eps>0$ there are unitaries $u_{AA'} \in \M_A\otimes M_d(\CC)\otimes 1$ of Alice and $u_{BB'}\in \M_B \otimes 1\otimes M_d(\CC)$ of Bob such that
\begin{align}\label{eq:SI_mbz}
    \norm{u_{AA'} u_{BB'}\big(\ket\Omega_{AB}\otimes \ket{0}_{A'}\ket{0}_{B'}\big) - \ket\Omega_{AB}\otimes\ket\Psi_{A'B'}} < \eps,
\end{align}
where $\ket{0}_{A'}\ket{0}_{B'}\in \CC^d\otimes \CC^d$ is a fixed but arbitrary product state. Note that while the resource system shared between Alice and Bob is described by general von Neumann algebras, the entangled states that are `embezzled' from this resource system are standard, finite-dimensional entangled states.

\begin{theorem}[{\cite{long_paper}}]
    A bipartite system $(\M,\M',\H)$ is a universal embezzler if and only if $\M$, and hence $\M'$, is a von Neumann algebra of type $\III_1$.
\end{theorem}

In the thermodynamic limit, the Hilbert space description of a many-body system is no longer unique, e.g., due to the presence or absence of infinitely many (quasi-)particles or different symmetry-broken ground states. To account for this fact, one separates the algebraic structure that is encoded in the commutativity of localized observables from the specification of a concrete Hilbert space representation.
This is mathematically captured by the notion of a \emph{quasi-local algebra}, which is the C*-algebra $\A$ that is generated by subalgebras $\A(\Lambda)$ describing observables localized in a finite region $\Lambda$ \cite{bratteli1997oa2}.
The Hilbert space enters by choice of a \emph{sector}, which is a representation of the quasi-local on a Hilbert space $\H$.
Typically, one is interested in ground state sectors of local Hamiltonians.
These arise as GNS representations of the quasi-local algebra $\A$ with respect to ground states of the dynamics generated by said Hamiltonians \cite[Prop.~5.3.19]{bratteli1997oa2}.
Different sectors can differ quite strongly in their physical properties.

To obtain a bipartite system in the sense of the above, we take two disjoint regions $\Lambda_A$ and $\Lambda_B$ and consider the pair of commuting subalgebras $\A(\Lambda_A), \A(\Lambda_B) \subset\A$.
In every sector, we get a pair of commuting von Neumann algebras $\M_A$ and $\M_B$ acting on the sector's Hilbert space $\H$ by taking the weak closures of $\A(\Lambda_A)$ and $\A(\Lambda_B)$ in the representation on $\H$.
If Haag duality holds, i.e., $\M_B=\M_A'$, which, in general, is a property of the sector, the triple $(\M_A,\M_B,\H)$ constitutes a bipartite system. 
Crucially, if the bipartite system $(\M_A,\M_B,\H)$ is a universal embezzler, the embezzling unitaries can be chosen from the local observable algebras $\A(\Lambda_A)$ and $\A(\Lambda_B)$:

\begin{lemma}\label{lem:cstar_mbz}
    Let $\A_A$ and $\A_B$ be C*-algebras, and let $\pi_A,\pi_B$ be a commuting pair of faithful representations on a joint Hilbert space $\H$.
    Let $\M_{A/B}$ be the weak closure of $\pi_{A/B}(\A_{A/B})$. Assume that Haag duality $\M_A=\M_B'$ holds and that $(\M_A,\M_B,\H)$ is a universal embezzler.
    
    For all state vectors $\ket\Omega_{AB}\in\H$ and $\ket\Psi_{A'B'}\in\CC^d\otimes\CC^d$ and all $\eps>0$, there exist unitaries $v_{AA'} \in \A_A\otimes M_d(\CC)$ and $v_{BB'}\in \A_B\otimes M_d(\CC) $ such that \eqref{eq:SI_mbz} holds with $u_{AA'} = (\pi_A\otimes \id \otimes1)(v_{AA'})$ and $u_{BB'}= (\pi_B\otimes 1\otimes \id)(v_{BB'})$.
\end{lemma}
\begin{proof}
    Since $(\M_A,\M_B,\H)$ is a universal embezzler, we find unitaries $u_{AA'}$ and $u_{BB'}$ in the weak closures of the respective algebras.
    The claim then follows from the general fact that the unitary group of a C*-algebra $\A$ on a Hilbert space $\H$ lies strongly dense in the unitary group of its weak closure $\overline \A^w$ \cite[Thm.~II.4.11]{takesaki1}.
\end{proof}

\section{Fermionic embezzlement}
\label{sec:fermi_mbz}
We want to claim that fermionic many-particle systems provide examples of universal embezzlers. To do this, we have to ensure that no problems arise from the parity superselection rule that constrains the implementable local unitary operations to those commuting with the parity operator (`even' unitaries). 
In this section, we show that, indeed, no problems arise as long as the finite-dimensional systems on which the embezzlement is performed are not fermionic themselves (it is, of course, impossible to embezzle odd fermionic states from even fermionic states).
To do this, we first review the general formalism of fermionic quantum systems and then present results that allow us to determine when fermionic systems can be interpreted as bipartite systems fulfilling Haag duality.

A fermionic quantum system is described by the C*-algebra $\CAR(\h)$ of canonical anticommutation relations (CAR) over a single-particle Hilbert space $\h$. The algebra $\CAR(\h)$ is generated by operators $a(\xi),a^\dagger(\eta)$ for $\ket\xi,\ket\eta\in\h$ such that
\begin{align}
    \{a(\xi),a(\eta)\} = \{a^\dagger(\xi),a^\dagger(\eta)\} = 0,\quad \{a(\xi),a^\dagger(\eta)\} = \braket{\xi}{\eta},
\end{align}
where $a^\dagger$ depends linearly on its argument and $a^\dagger(\eta) = a(\eta)^*$.\footnote{The C*-norm on on the CAR algebra is unique \cite[Sec.~6]{evans1998qsym}.}
The operator $-1$ on $\h$ induces an automorphism $\gamma$ of order $2$ on $\CAR(\h)$, which is called the parity automorphism.
Every $x\in\CAR(\h)$ uniquely decomposes into an even and odd part: $x=\frac12(x+\gamma(x)) + \frac12(x-\gamma(x))= x_{e}+x_{o}$, where $\gamma(x_{e})=x_{e}$ and $\gamma(x_{o})= -x_{o}$ under the parity automorphism.
We define the (unital) even *-subalgebra of $\CAR(\h)$ as
\begin{align}
    \CAR_{e}(\h) := \{ x\in \CAR(\h)\ :\ x=\gamma(x) \}. 
\end{align}
Due to the parity super-selection rule, valid physical observables are represented by $\CAR_{e}(\h)$ \cite{wick_intrinsic_1952}.

Localization regions of operators are determined by subspaces $q\h$ of the single-particle Hilbert space $\h$, where $q$ is a projection operator. 
For example, $\h$ could be given by $\ell^2(\ZZ)$ and $q\h = \ell^2(\ZZ_+)$ describes the operators associated to the right half-chain.
The operators localized on the right half-chain are then given by $\CAR(q\h)\subseteq \CAR(\h)$. 
The local \emph{observables} (relative to $q$) are given by the even subalgebra $\CAR_{e}(q\h)\subset \CAR_{e}(\h)$. 

A state $\omega$ on $\CAR(\h)$ is called \emph{even} if $\omega\circ\gamma=\omega$.
If $\omega$ is an even state, we will denote by $\omega
_{e}$ its restriction to $\CAR_{e}(\h)$. It follows that an even state $\omega$ is pure if and only if $\omega_{e}$ is pure \cite[Lem.~6.23]{evans1998qsym}.

In the following, $\omega$ will be an even pure state.
We denote the GNS representations of $\omega$ and $\omega_{e}$ by $(\pi,\H,\ket\Omega)$ and $(\pi_{e},\H_{e},\ket{\Omega_{e}})$, respectively.
The GNS representation of $\omega_{e}$ arises by truncating the restriction of $(\pi,\H,\ket\Omega)$ to $\CAR_{e}(\h)$ with the projection $P_e=[\pi(\CAR_{e}(\h))\ket\Omega]$, i.e., $\H_{e}=P_{e}\H$, $\ket{\Omega_{e}}=P_{e}{\ket\Omega}$ and $\pi_{e}=P_{e}\pi(\placeholder)P_{e}$.
The parity automorphism $\gamma$ is implemented by the self-adjoint unitary
\begin{align}
    (-1)^F = P_{e} - (1-P_{e}) = 2 P_{e} - 1.
\end{align}
We also define the \emph{twist operator} \cite{bisognano1976duality,baumgaertel2002twisted_duality}
\begin{align}
    Z = P_{e} - i (1-P_{e}) = (1+i)P_{e} - 1 = \tfrac{1}{1+i}(1+i(-1)^{F}).
\end{align}
which is unitary, fulfills $Z^2 = (-1)^F$ and commutes with even operators.

For a projection $q$ on $\h$, we define the von Neumann algebras 
\begin{align}
    \M(q) :=\pi_{e}(\CAR_{e}(q\h))''
\end{align}
and
\begin{align}
    \F(q) := \pi(\CAR(q\h))'', \qquad  \F_e(q) := \pi(\CAR_{e}(q\h))''.
\end{align}
acting on $\H_{e}$ and $\H$ respectively. Since $\omega$ is a pure state on $\CAR(\h)$ and $\CAR(q\h)$ and $\CAR(q^\perp\h)$ generate $\CAR(\h)$, we have $\F(q)\vee\F(q^\perp)=\B(\H)$.

Local operations (relative to $q$) are precisely those with Kraus operators $k_j \in \M(q)$ \cite{locc}. In particular, a local unitary operation is represented by a unitary $u\in \M(q)$. 
Our goal in the following is to determine when we have Haag duality, i.e., $\M(q) = \M(q^\perp)'$, when $\M(q)$ is a factor, and when $\M(q)$ has type $\III_1$. 
In the following, we need to be careful to distinguish between $\M(q)$ and $\F_e(q)$.

Note that $\F(q)$ can never be the commutant of $\F(q^\perp)$ because the odd operators in $\F(q)$ anticommute with the odd operators in $\F(q^\perp)$. The importance of the twist operator stems from the relations $Z x_{o} Z^{*}  = i x_{o} (-1)^F$ and $Z x_{e} Z^{*} = x_{e}$. 
Hence, if $y_{o} x_{o} = -x_{o} y_{o}$ we have
\begin{align}
    y_{o} Z x_{o} Z^{*} = Z x_{o} Z^{*} y_{o}. 
\end{align}
This shows that $\F(q)$ commutes with $Z \F(q^\perp)Z^{*}$, i.e.,
\begin{align}
\label{eq:twist_inclusion}
    Z\F(q^\perp)Z^{*} & \subseteq \F(q)'.
\end{align}

\begin{lemma}
\label{lem:pure_factor}
If $\omega$ is an even pure state, then $\F(q)$ is a factor.
\end{lemma}
\begin{proof}
    Since $\omega$ is pure we have $\B(\H) = \F(q)\vee \F(q^\perp)$ (because $\CAR(q\h)$ and $\CAR(q^\perp\h)$ generate $\CAR(\h)$).
    Taking commutants and using \eqref{eq:twist_inclusion}, we find
    \begin{align}
        \CC 1 = \F(q)' \cap \F(q^\perp)' \supseteq \F(q)' \cap Z \F(q) Z^{*}.
    \end{align}
    Now suppose $x\in \F(q)\cap \F(q)'=Z(\F(q))$. If $x$ is even, it is also part of $Z \F(q) Z^{*}$ and therefore $x\propto 1$. If $x\neq0$ is odd, then $x^2$, $x^*x$ and $xx^*$ are even elements of $Z(\F(q))$ and hence $x^2\propto 1, x^*x\propto 1, xx^*\propto 1$. Since $x\neq0$, we have $x^*x=c 1$ for some $c>0$, and we can rescale $x$ such that $x^*x=1$. This implies $xx^*=1$ and that $x$ is unitary, because $xx^*\propto 1$. Now, $\norm{x^2}^2=\norm{(x^*)^2x^2}=1$ yields $x^2=c'1$ with $|c'|=1$, and we have $x=x^*$ and $x^2=1$ after another rescaling. Then, $q=\tfrac{1}{2}(1-x)\in Z(\F(q))$ is a projection. Since a projection is necessarily even ($q=q^2$), this yields $q=0$ and, therefore, $x=1$. This is a contradiction since $x$ was assumed to be odd. Thus, $\F(q)\cap \F(q)' = \CC 1$.
\end{proof}
Equality in \eqref{eq:twist_inclusion} can be understood as a fermionic version of Haag duality.
\begin{definition}
We say that $\omega$ and a closed subspace $q\h\subset \h$ fulfill \emph{twisted duality} if
\begin{align}
    \F(q)' = Z \F(q^\perp)Z^{*}.
\end{align}
\end{definition}
While the operational significance of twisted duality may at first be unclear, we show that it implies Haag duality on the level of the physical observables.
\begin{proposition}\label{prop:duality}
    Suppose $\omega$ is an even pure state on $\CAR(\h)$. If $\omega$ and $q\h\subset \h$ fulfill twisted duality, then $\M(q)'=\M(q^\perp)$.
\end{proposition}
The proof of \cref{prop:duality} makes use of the following Lemma, which is a reformulation of \cite[Prop.~17.64]{derezinski_mathematics_2023}.
\begin{lemma}\label{lem:even-commutant}
    Let $\omega$ be pure and even, and suppose that $\omega$ and $q$ fulfill twisted duality. Then $ \F_e(q)' = \F(q^\perp) \vee \{(-1)^F\}''$.
\end{lemma}
\begin{proof}
The even part $\F_e(q)$ can be written as $\F_e(q) = \F(q) \cap \{(-1)^F\}'$.
We now write
\begin{align} \nonumber
    \F(q^\perp)' \cap \{(-1)^F\}' & = Z \F(q)Z^{*} \cap \{(-1)^F\}'\nonumber\\
    & = Z \left(\F(q)\cap \{(-1)^F\}'\right)Z^{*} \nonumber\\
    &= Z  \F_e(q) Z^{*} =  \F_e(q),
\end{align}
where we used that $Z$ commutes with even operators. Taking the commutant on both sides yields the claim.
\end{proof}
\begin{proof}[Proof of \cref{prop:duality}]
We have that $\M(q) =  \F_e(q) P_{e}$ with $P_{e} \in  \F_e(q)'$. By \cite[Prop.~5.5.6]{KadisonRingrose1} this implies that $\M(q)' = P_{e}  \F_e(q)' P_{e}$. Using \cref{lem:even-commutant} we find
\begin{align}
    \M(q)' = P_{e} \F(q^\perp)\vee\{(-1)^F\}'' P_{e} = P_{e} \F(q^\perp) P_{e} = \M(q^\perp),
\end{align}
where we used that $P_{e}=\tfrac{1}{2}(1+(-1)^F)\in \{(-1)^F\}''$.
\end{proof}

\begin{remark}[Jordan-Wigner transformations and Haag duality]\label{remark:jw} It is well known that one-dimen\-sional fermionic systems have a dual representation in terms of spin chains via the Jordan-Wigner transformation. For a bipartition into two semi-infinite half-chains, the Jordan-Wigner transformation preserves the structure of the subsystems. Due to a lack of parity super-selection rule, the physical observable algebras are now given by the Jordan-Wigner duals of $\F(q)$ and $\F(q^\perp)$. If the fermionic factors $\F(q)$ arise from a pure quasi-free state, twisted duality will imply Haag duality for the respective local factors describing the spin system. For details of this construction, we refer to \cite{keyl2006} (see also \cite[Sec.~6.10]{evans1998qsym}). As a consequence, our results also apply to the Jordan-Wigner duals of the considered fermionic systems. 
\end{remark}

Next, we need to determine the type of $\M(q)$. 
The parity automorphism $\gamma$, represented on $\H$ by $(-1)^F$, restricts to an automorphism $\gamma|_q$ on $\F(q)$.
Depending on $\omega$, this restricted automorphism may be an inner automorphism of $\F(q)$, i.e., it may be represented as $\gamma|_q(x) = U_q x U_q^*$ by a unitary operator $U_q \in \F(q)$, or not. That $\gamma|_q$ is inner on $\F(q)$ does not imply that it is inner as an automorphism on $\CAR(q\h)$. We, therefore, call $\gamma|_q$ \emph{weakly inner} if it is inner on $\F(q)$. If $\gamma|_q$ is not weakly inner, we say that it is \emph{outer}.

\begin{proposition}\label{prop:factor-type}
Let $\omega$ be pure and even, and suppose that $\omega$ and $q$ fulfill twisted duality. Then:
\begin{enumerate}
    \item\label{item:cutdown}
    The cutdown $\F_e(q)\ni x\mapsto P_e x P_e \in \M(q)$ is an isomorphism.
    \item\label{item:factor} $\F_e(q)\cong\M(q)$ is a factor if and only if $\gamma|_q$ is outer.
    \item\label{item:types} If $\gamma|_q$ is weakly inner, the (sub)type of the von Neumann algebra $\F_e(q)'$ coincides with the (sub)type of $\F(q^{\perp})$.
    If $\gamma|_q$ is outer and $\F(q^{\perp})$ has type $X\in \{\II_1,\II_\infty,\III_0,\III_1\}$ then $\F_e(q)'$ has type $X$.
    If $\gamma|_q$ is outer and $\F(q^{\perp})$ has type $\III_\lambda$ with $0<\lambda<1$, then $\F_e(q)'$ has type $\III_{\lambda^2},\III_{\lambda^{1/2}}$, or $\III_\lambda$. 
\end{enumerate}
\end{proposition}
The proof of \cref{prop:factor-type} relies on the following Lemma.
\begin{lemma}\label{lem:local-parity} Let $\omega$ be even, and suppose that $\omega$ and $q\h$ fulfill twisted duality. Then $\gamma|_{q}$ is outer if and only if $\gamma|_{q^{\perp}}$ is outer.
\end{lemma}
\begin{proof}
We show that if $\gamma|_q$ is weakly inner, then  $\gamma|_{q^\perp}$ is weakly inner. The claim then follows because twisted duality for $q$ is equivalent to twisted duality for $q^\perp$.
Suppose $\gamma|_q$ is implemented by a unitary $U_q\in \F(q)$. This implies $(-1)^{F}U_{q}(-1)^{F} = \gamma(U_{q}) = \gamma|_{q}(U_{q}) = U_{q}U_{q}U_{q}^{*} = U_{q}$.Thus, $U_{q}\in\F_{e}(q)$ is even. Then, $U_{q^\perp}:=U_q (-1)^F\in \F(q)'$ is a unitary that implements $\gamma|_{q^{\perp}}$. 
We observe that $U_{q^\perp}$ is even since $U_q$ is even. Thus, $ U_{q^\perp} = Z^{*}U_{q^\perp}Z\in\F_e(q^{\perp})$.
\end{proof}

\begin{proof}[Proof of \cref{prop:factor-type}]
We first prove items \ref{item:cutdown} -- \ref{item:types} in the case where $\gamma_q$ is weakly inner and make use of the notation in the proof of \cref{lem:local-parity}.  In this case $1\neq U_q \in \F_e(q)\cap\F_e(q)'$ and, therefore, $\F_e(q)$ is not a factor.
By \cref{lem:local-parity}, we can write $(-1)^F = U_q U_{q^\perp}$ and, by \cref{lem:even-commutant}, we find $\F_e(q)' = \F(q^\perp) \vee \{U_q\}''$.
Hence, the center of $\F_e(q)$ is generated by $U_q$.
Now, we note that $U_{q}$ can be chosen self-adjoint, i.e., $U_{q}^{2}=1$: $(\gamma|_q)^2=\id$ implies that $U_{q}^{2}$ is central, i.e., $U_{q}^{2}\in\F(q)\cap\F(q)'$. Therefore, we can find a central unitary $V\in\F(q)\cap\F(q)'$ such that $U_{q}^{2}=V^{2}$ and replace $U_{q}$ by $\tilde{U}_{q}=V^{*}U_{q}$ which is self-adjoint.
Since $U_q$ is a self-adjoint unitary, its spectrum is $\spec U_q=\{\pm1\}$ with spectral projections $P_{q}^{\pm} = \tfrac{1}{2}(1\pm U_q)$. This implies $\F_e(q)' \cong \F(q^\perp) \oplus \F(q^\perp)$ and, by duality, $\F_e(q^{\perp})' \cong \F(q) \oplus \F(q)$. This shows the first part of item~\ref{item:types}.
To infer that $\F_e(q) \cong \M(q)$, we show $P_{e} = \tfrac{1}{2}(1+(-1)^{F})\in\F_e(q)'$ has central cover $Z_{e}=1$.
Since $(-1)^F= U_qU_{q^\perp}$, we find
\begin{align*}
P^{\pm}_{q}P_{e} = \tfrac{1}{4}(1+U_qU_{q^\perp}\pm(U_q+U_{q^\perp})) = \tfrac{1}{2}(1\pm U_q)\tfrac{1}{2}(1\pm U_{q^\perp}) = P_{q}^{\pm}P_{q^{\perp}}^{\pm},
\end{align*}
meaning that none of the two minimal projections in the center of $\F_e(q)'$ covers $P_{e}$. Therefore, we have $Z_{e}=1$ and it follows from \cite[Prop.~5.5.6]{KadisonRingrose1} that $\F_e(q)$ and $\M(q)$ are isomorphic.

Now suppose that $\gamma|_q$ is outer. 
Since $(-1)^F$ implements the outer automorphism $\gamma|_{q^\perp}$ on $\F(q^\perp)$, the relation $\F_e(q)' = \F(q^\perp) \vee \{(-1)^F\}''$ implies that $\F_e(q)'$ is (isomorphic to) the crossed product of $\F(q^\perp)$ by an outer action of $\ZZ_2$ \cite[Ex.~X.1.1]{takesaki2}. It follows from \cite[Prop.~13.1.5]{KadisonRingrose2} that $\F_e(q)'$ is a factor and hence $\F_e(q)$ is a factor as well.
Since $P_{e}$ is in the commutant of $ \F_e(q)$, i.e, $P_{e}\equiv[\F_e(q)]\in\F_e(q)'$, we have that $\M(q) = P_{e}  \F_e(q) P_{e} =  \F_e(q) P_{e} \cong \F_e(q)Z_{e}$, where $Z_{e}$ is the central cover of $P_{e}$ in $\F_e(q)'$ (as before). Since $\F_e(q)$ is a factor, we have $Z_{e}=1$ implying that $\F_e(q)\cong \M(q)$ and that $\M(q)$ is a factor (again by \cite[Prop.~5.5.6]{KadisonRingrose1}).
This shows  item~\ref{item:factor} and the first part of item~\ref{item:types}.

Consider the inclusion $\F(q^{\perp}) \subset \F_e(q)'=\F(q^{\perp})\vee\{(-1)^{F}\}''$.
Due to the crossed product structure, there is a conditional expectation $\F_e(q)' \rightarrow \F(q^{\perp})$ of index $2$ \cite[Ex.~1.1]{loi_theory_1992}. 
If $\F(q^{\perp})$ has type $\III$, the assertion then follows from  \cite[Thm.~2.7]{loi_theory_1992} and \cite[Cor.~2.11]{loi_theory_1992}.
If $\F(q^{\perp})$ has type $\II_1$, the assertion follows from \cite[Ex.~5.3]{evans1998qsym}.
If $\F(q^{\perp})$ has type $\II_{\infty}$, the assertion follows from \cite[Prop.~2.3]{longo1989index1} and the type $\II_{1}$ case. 
\end{proof}

\begin{corollary}
\label{cor:factor-type}
Given the same assumptions as in \cref{prop:factor-type}, if $\F(q)$ is of type $\III$, then the subtype of $\F_{e}(q)$ coincides with the subtype of $\F(q)$.
\end{corollary}
\begin{proof}
Since $\F(q)$ is of type $\III$, we know that $\F(q)$ is in standard form \cite[III.2.6.16]{blackadar_operator_2006} and, thus, anti-isomorphic to $\F(q^{\perp})$ by twisted duality. By \cref{prop:factor-type}, we know that $\F_{e}(q)'$ is of type $\III$ and anti-isomorphic to $\F_{e}(q)$ as it is also in standard form. 
\end{proof}

\begin{remark}\label{rem:factor-type}
In view of \cref{prop:factor-type}, we note that $\gamma|_{q}$ is always weakly inner if $q$ is finite dimensional ($\dim\textup{ran}(q)<\infty$). This implies that the algebraic structure of the fermionic observable algebras can vastly differ between finite and infinite dimensions: $\M(q)$ always has a center in the former case, while it can be a factor in the latter case.
This potential structural change has direct physical consequences. If the observable algebra $\M(q)$ has a center, classical information can be extracted without perturbing the system. This classical information is precisely the parity information of the given state.
\end{remark}

\subsection{Gauge-invariant quasi-free states}
We briefly collect some general properties of gauge-invariant quasi-free states that we require; see, for example, \cite{powers1970free_states}.
Since we only consider gauge-invariant states, we write quasi-free states instead of gauge-invariant quasi-free states in the following. We refer to \cite{balslev_states_1968,araki1970car_bog,evans1998qsym} for discussions of general quasi-free states.
If $0\leq s\leq 1$ is an operator on $\h$, the quasi-free state $\omega_s$ on $\CAR(\h)$ associated to $s$ is determined by
\begin{align}\label{eq:quasi-free}
    \omega_s(a(\xi_1)\cdots a(\xi_m) a^\dagger(\eta_n) \cdots a^\dagger(\eta_1)) = \delta_{mn} \det([\bra{\xi_j}s \ket{\eta_k}]_{j,k=1}^m),\quad \forall \ket{\eta_j},\ket{\xi_k}\in \h. 
\end{align}
It follows that $\omega_s$ is completely determined by the second moments $\omega_s(a(\xi)a^\dagger(\eta))= \bra \xi s\ket\eta$ and hence by $s$. Eq.~\eqref{eq:quasi-free} thus expresses Wick's theorem. 
Every quasi-free state is even, and $\omega_s$ is pure if and only if $s$ is a projection operator. In this case, $\omega_s$ is called a \emph{Fock state} because its GNS Hilbert space can be identified with a Fock space. 
Note that if $p$ is a projector we have $\omega_p(a^\dagger(\xi)a(\xi)) = 0$ whenever $\xi \in p\h$. In particular, if $p$ is the spectral projection onto the positive part of the single-particle Hamiltonian $h$ on $\h$, then $\omega_p$ is a ground state. This follows from \cite[Prop.~5.3.19.(5)]{bratteli1997oa2} since, by construction, $\omega_p$ is invariant under the quasi-free dynamics generated by $h$, the generator $H_{p}$ of said dynamics in the Fock representation $(\pi_{p},\H_{p})$ induced by $\omega_p$ is evidently positive, and $e^{itH_{p}}\in\pi_{p}(\CAR(\h))''=\B(\H_{p})$ by irreducibility.
It has been shown in \cite{araki1964vna_free_field,foit1983twisted_duality,derezinski_mathematics_2023} that quasi-free fermionic Fock states fulfill twisted duality. 
We can, therefore, summarize the results of the previous section relevant to us as:
\begin{corollary}
If $\omega_p$ is a quasi-free Fock state on $\CAR(\h)$, then $\M(q)'=\M(q^\perp)$. Hence the triple $(\M(q),\M(q^\perp),\H_{e})$ defines a bipartite system. If $\F(q)$ has type $\III_1$, this bipartite system is a universal embezzler.
\end{corollary}

Furthermore, whenever $(\M(q),\M(q^\perp),\H_e)$ is a universal embezzler, the embezzling unitaries can be chosen from the even CAR algebras, i.e., $\CAR_{e}(q\h)$ and $\CAR_{e}(q^{\perp}\h)$. 
This follows from \cref{lem:cstar_mbz} because $\CAR_e(\h)$ is simple if $\h$ is infinite dimensional, which is true for $\CAR_{e}(q\h)$ and $\CAR_{e}(q^{\perp}\h)$ \cite{doplicher_simplicity_1968}. Hence the representation $\pi_e$ is faithful and every unitary $u\in \pi(\CAR_e(q\h))$ corresponds to a unitary in $u_0\in \CAR_e(q\h)$ via $u=\pi_e(u_0)$.
The representation $\pi$ of $\CAR(q\h)$ is faithful as well. In fact, $\CAR(q\h)$ and $\CAR_{e}(q\h)$ are *-isomorphic for infinite-dimensional $\h$ \cite{stormer_even_1970, binnenhei997on_even_car} (cp.~\cref{rem:factor-type}). Since infinite tensor products of matrix algebras are simple, we can draw the same conclusion for spin systems and states that are dual to quasi-free Fock states via the Jordan-Wigner transformation (cf.~\cref{remark:jw}).

In view of \cref{prop:factor-type} and \cref{rem:factor-type}, it is natural to ask whether $\M(q)$ is always a factor if it has type $\III$, i.e., whether $\gamma|_q$ is always outer in this case. We show a stronger statement:
\begin{proposition}\label{prop:inner-to-I}
    Consider a quasi-free Fock state $\omega_p$ on $\CAR(\h)$. Then $\gamma|_q$  is implementable by a unitary operator on $\H$ if and only if $\F(q)$ has type $\I$. 
\end{proposition}
Since all automorphisms of type $\I$ factors are inner, we find:
\begin{corollary}
    Consider a quasi-free Fock state $\omega_p$ on $\CAR(\h)$. The following are equivalent:
    \begin{enumerate}
        \item $\gamma|_q$ is weakly inner.
        \item $\F(q)$ has type $\I$.
        \item $\M(q)$ has type $\I$.
        \item $\M(q)$ is not a factor.
    \end{enumerate}
\end{corollary}
\begin{proof}
The equivalence of items 1 and 2 is shown in \cref{prop:inner-to-I}. That $\F(q)$ and $\M(q)$ have the same types follows from \cref{prop:factor-type}, because $\F(q)$ has type $\I$ if and only if $\F(q^\perp)$ has type $\I$ by \cref{lem:local-parity}. By \cref{prop:factor-type} $\F_e(q)'$ has type $\I$ if and only if $\F(q^\perp)$ has type $\I$. Since a von Neumann algebra and its commutant always have the same type, we find that $\M(q)\cong \F_e(q)$ has type $\I$ if and only if $\F_e(q)'$ has type $\I$. The equivalence between items 1 and 4 is item \ref{item:factor} in \cref{prop:factor-type}.    
\end{proof}

We emphasize that \cref{prop:inner-to-I} shows that if $\M(q)$ is not of type $\I$, then $\gamma|_q$ cannot even be implemented by a unitary operator on $\H$ that acts globally. 
In this case, there is neither a local nor a global operator on $\H$ that measures the local parity of Alice's or Bob's system.

To show the statement, we make use of the following Lemma:
\begin{lemma}\label{lem:implementability}  Consider a quasi-free Fock state $\omega_p$ on $\CAR(\h)$. $\gamma_q$ is implemented by a unitary $U_q$ on $\H$ if and only if $qpq^{\perp}$ is Hilbert-Schmidt.
\end{lemma}
\begin{proof}
The state $\omega_p \circ \gamma|_q$ is a quasi-free pure state represented by the covariance $s = u_q p u_q$, with the unitary $u_q  = 1-2q = q^{\perp}-q$. 
According to \cite[Lem.~9.4]{araki1970car_bog}, $\gamma|_{q}$ is unitarily implementable if and only if $p-u_{q}pu_{q}$ is Hilbert-Schmidt.
Since
\begin{align}
    p - (1-2q)p(1-2q) = 2(qpq^\perp + q^\perp p q) 
\end{align}
this is true if and only if $q p q^\perp$ is Hilbert-Schmidt. 
\end{proof}
\begin{lemma}\label{lem:HS-to-T}
    Let $p,q$ be projections on $\h$. Then $qpq^\perp$ is Hilbert-Schmidt if and only if $qpq = e +t$, where $e$ is a spectral projection of $qpq$ and $t$ is trace-class.
\end{lemma}
\begin{proof}
By the theory of two projections \cite{halmos_two_1969} the Hilbert space $\h$ can be decomposed as $\h_0\oplus \h_1$, so that 
$p$ and $q$ decompose into orthogonal projections as $p=p_0\oplus p_1$ and $q=q_0\oplus q_1$ with  $[p_0,q_0]=0$, $[p_1,q_1]\neq 0$ and $p_1,q_1$ are unitarily equivalent to projection operators of the form
\begin{align}
    q_1\cong \begin{pmatrix}1 & 0\\ 0& 0\end{pmatrix},\quad p_1 \cong \begin{pmatrix} c^2 & c(1-c^2)^{1/2} \\ c(1-c^2)^{1/2} & 1-c^2\end{pmatrix} 
\end{align}
on $\h_1\cong\mathfrak k\oplus\mathfrak k$, where $0\leq c\leq 1$. 
Thus $qpq = q_0p_0 \oplus q_1 p_1 q_1$. Absorbing the projection $q_0p_0$ into $e$, we can thus assume without loss of generality that $q$ and $p$ are in general position, i.e., $p=p_1$, $q=q_1$ and $qpq=c^2$.
If $c^2 = e+t$ with $t$ trace-class, then
\begin{align}
\tr((qpq^\perp)^*qpq^\perp) = \tr(c^2(1-c^2)) = \tr((e+t)(e^\perp - t)) = \tr(te^\perp - et - t^2) <\infty.
\end{align}
Hence $qpq^\perp$ is Hilbert-Schmidt. 
Conversely, let $e = \chi_{[1/2,1]}(c^2)$ and assume that $qpq^\perp = c\sqrt{1-c^2}$ is Hilbert-Schmidt. We then have
\begin{align}
    \tr(c^2(1-c^2)) = \tr(e c^2(1-c^2)) + \tr(e^\perp c^2(1-c^2)) <\infty.
\end{align}
Since $ec^2\geq 1/2 e$ and $e^\perp (1-c^2) \geq 1/2 e^\perp$ we have
\begin{align}
    \tr (e(1-c^2)) \leq 2 \tr(e c^2(1-c^2)) <\infty. 
\end{align}
and
\begin{align}
    \tr (e^\perp c^2) \leq 2 \tr(e^\perp c^2(1-c^2)) <\infty. 
\end{align}
Setting $t:= c^2 - e$,  we find
\begin{align}
    \tr(|t|)  = \tr(e(1-c^2)) + \tr(e^\perp c^2)  < \infty.
\end{align}
    \end{proof}
\begin{proof}[Proof of \cref{prop:inner-to-I}]
    By \cite[Lem.~5.3]{powers1970free_states}, a quasi-free state $\omega_s$ on $\CAR(\h)$ yields a type $\I$ factors $\F(\h)$ if and only if $s= e+t$, where $e$ is a spectral projection of $s$ and $t$ is trace-class. By assumption, $\gamma|_q$ is weakly inner. Therefore, by \cref{lem:implementability}, we find that $qpq^\perp$ is Hilbert-Schmidt, and hence \cref{lem:HS-to-T} implies that $\F(q)$ type $\I$. 
\end{proof}

\section{Fermionic many-particle systems as universal embezzlers}
\label{sec:fermin_mbz_univ} 
We consider the fermionic second quantization $H$ of a (self-adjoint) single-particle Hamiltonian $h$ on $\mathfrak{h}=\ell^2(\ZZ)\otimes \CC^b$ with standard basis $\{\ket{x}\ox \ket{l}\}_{x\in\ZZ, l=1,...,b}$. 
\begin{equation}
    H = \sum_{x,y}\sum_{l,m} h_{l,m}(y-x)\,a^{\dag}_l(y)\,a_m(x), \qquad h_{lm}(y-x) = \bra{y}\ox\bra{l}h\ket{x}\ox\ket{m},
\end{equation}
where $\{a_l(x),a^{\dag}_{m}(y)\}=\delta_{l,m}\delta_{x,y}$ realize the canonical anti-commutation relations with respect to $\mathfrak{h}$ and we assumed that $h$ is translation invariant with respect to the shift on $\ZZ$. We denote the resulting algebra of canonical anti-commutation relations by $\CAR(\h)$. 

We assume that the functions $x\mapsto h_{lm}(x)$ have sufficiently fast decay, e.g., absolute summability, to ensure that their Fourier transforms $\hat h_{lm}(k)$ are continuous functions on the unit circle $S^1 \cong [0,2\pi)$ \cite{katznelson_introduction_2004}. 
For each $k\in S^1$ the matrix $\hat h(k) = [\hat h_{lm}(k)]$ is self-adjoint. The one-particle Hamiltonian $h$ is hence unitarily equivalent to a multiplication operator on  $L^2(S^1,\CC^b)\cong \ell^2(\ZZ)\otimes \CC^b$. Note that $h$ is necessarily bounded:
\begin{align}
    \norm{h} = \norm{\hat h} \leq \sup_{\substack{\ket\psi \in L^2(S^1,\CC^b) \\  \braket{\psi}{\psi}=1}} \int_{S^1} |\langle \psi(k)| \hat h(k)|\psi(k)\rangle| dk \leq \int_{S^1}\norm{\hat h(k)} dk < \infty.
\end{align}

A multiplication operator $\hat f\in L^\infty(S^1, M_n(\CC))$ is said to be discontinuous at $k\in S^1$ if every representative is discontinuous at $k$, i.e., the left- and right-limits do not coincide. 
Furthermore, $\hat f$ is piecewise continuous if it only has finitely many points of discontinuity.
From now on, we make the following regularity assumption: 

\begin{assumption}\label{assumption}
    The symbol $\hat h:S^1\to M_n(\CC)$ is such that $\hat p_+ = \chi_{(0,\oo)}(\hat h)\in L^\oo(S^1,M_b(\CC))$ is piecewise continuous.
\end{assumption}

\begin{definition}
    A Hamiltonian $H$ satisfying \cref{assumption} is called \emph{critical} if $\hat p_+\in L^\oo(S^1,M_b(\CC))$ has at least one point of discontinuity.
\end{definition}

We can always choose $\hat p_+$ to be represented by a matrix-valued function that is continuous from the left. The left and right limits of $\hat p_+$ at a discontinuity are distinct projectors.

\begin{example}
    The Hamiltonian $H$ induced by $\hat h(k) = \mathrm{diag}(1+\cos(k), - (1+\cos(k))$ is not critical in our sense, since $\hat p_+$ is constant equal to $\mathrm{diag}(1,0)$. This is despite the fact that $\hat h(\pi)=0$ and that $H$ is gapless.  
\end{example}
\begin{example}
The Hamiltonian $H_{\text{XX}} = \sum_x \big(a(x)^\dagger a(x+1) + a^\dagger(x+1)a(x)\big)$ leads to a scalar symbol $\hat h(k) = 2\cos(k)$. The projection $\hat p_+$ onto the strictly positive part of $\hat h$ is given by $\hat p_+ = 1-\chi_{[\pi/2,3\pi/2]}$. It is simply an indicator function with discontinuities at $\pi/2$ and $3\pi/2$. Hence, the Hamiltonian is critical. 
\end{example}

\begin{example}[Su-Schrieffer-Heeger model \cite{su_soliton_1980}] The Hamiltonian 
\begin{align}
    H_{\text{SSH}} = \sum_{x} \left(a_1^\dagger(x)a_2(x) + a_2(x)^\dagger a_1(x) + a_1^\dagger(x)a_2(x+1) + a_2(x+1)^\dagger a_1(x)\right),
\end{align}
leads to the symbol
\begin{align}
    \hat h (k) = \begin{pmatrix}
        0 & 1+ e^{i k}\\
        1+ e^{-ik} &0.
\end{pmatrix}
\end{align}
Clearly, $\hat h$ is continuous on $S^1$. 
Its eigenvalues can be chosen analytically up to the boundary point $p=2\pi$  as $\lambda_\pm(k) =\pm 2\cos(k/2)$, entailing $\lambda_+(k+2\pi) = \lambda_-(k)$, with $\lambda_+(k)\geq 0$ for $k\in[0,\pi]$ and $\lambda_-(k)\geq 0$ for $k\in[\pi,2\pi]$. Alternatively, one could choose them as continuous but non-smooth functions of $k\in S^1$.
However, its eigenprojectors are  discontinuous at $k=\pi$ since the eigenvectors $\varphi_\pm(k)\in \CC^2$ of $\hat h(k)$ corresponding to $\lambda_\pm$ fulfill $\varphi_+(k+2\pi) = \varphi_-(k)$, where $\varphi_\pm(k) = (\pm e^{i k/2},1)^\top$. The topology of the band structure is that of a Möbius band.
The projection $\hat p_+$ on the strictly positive part of $\hat h$ has a discontinuity at $p=\pi$, since $\lim_{\eps\to 0}\varphi_+(\pi-\eps) = (i,1)^\top$, but $\lim_{\eps\to 0}\varphi_-(\pi+\eps) = (-i,1)^\top$.
\end{example}

Given a pure, quasi-free ground state $\omega_p$ of $H$, we denote by $q$ the projection onto the subspace $\ell^2(\ZZ_+)\otimes \CC^b \subset \ell^2(\ZZ)\otimes \CC^b=\h$ corresponding to the right half-chain and define $\M_A := \M(q)$, $\M_B:=\M(q^\perp)=\M_A'$ in the notation of section~\ref{sec:fermi_mbz}. $(\M_A,\M_B,\H_e)$ defines the bipartite system describing the physically relevant, even observable algebras.
The state $\omega_p$ is represented by a vector $\ket \Omega\in \H_e$. 

\begin{theorem}\label{thm:main-fermions}
If $H$ fulfills assumption \ref{assumption} and is critical, there exists a gauge-invariant quasi-free, pure, and translation-invariant ground state such that $\M_A$ is a type $\III_1$ factor. 
\end{theorem}
The proof of the theorem proceeds via several steps:
\begin{enumerate}
    \item A quasi-free, translation-invariant, and pure ground state is represented by a projection $p = p_+ + p_0$, where $p_+$ is the spectral projection onto the strictly positive part of $h$ and $p_0$ is a projection onto a subspace in the kernel of $h$ that commutes with translations. By translation-invariance, $p$ is represented by a projection valued symbol $\hat p = \hat p_+ + \hat p_0\in L^\infty(S^1,M_b(\CC))$ when viewed in Fourier space, where $\hat p_+(k)$ is the spectral projection onto the strictly positive part of $\hat h(k)$ and $\hat p_0(k)$ is a projection onto a subspace in the kernel of $\hat h(k)$. We call $\hat p$ the \emph{symbol} of $p$. Clearly $\hat p_+ \hat p_0 = \hat p_0\hat p_+ = 0$. 
    By assumption, $\hat p_+$ is piecewise continuous with finitely many discontinuities.
    The regularity of $p_0$ is, in principle, arbitrary. In the following, we therefore choose $p_0=0$, but our arguments apply whenever $\hat p_0$ is piecewise continuous. 

    \item We denote by $q$ the inclusion operator $\ell^2(\ZZ_+)\ox\CC^{b} \hookrightarrow \ell^2(\ZZ)\ox\CC^{b}$.
    Then, the state associated with the right half-chain is given by the quasi-free state $\omega_{q^* p q}$. 
Note that $qq^*$ is the projection onto $\ell^2(\ZZ_+)\ox\CC^{b}$.
    
    \item By \cref{lem:pure_factor}, the half-chain von Neumann algebra $\F(q q^*)$ is a factor. It is isomorphic to the factor constructed from $\omega_{q^* p q}$ and $\CAR(\ell^2(\ZZ_+)\ox \CC^b)$ by \cite[Prop.~5.5.6]{KadisonRingrose1} because the latter is the cut-down of $\F(qq^*)$ by the projection onto the invariant subspace $[\F(qq^*)\H)]$.
    In the following we identify $q$ with $qq^*$, so that $\F(q) = \F(qq^*)$.

    \item\label{item:classification} Our aim is to show that $\F(q)$ is a factor type $\III_1$. The previous point implies that it is sufficient to show that  $\spec(q^*pq)$, the spectrum of $q^*pq$, contains a non-trivial interval, see section~\ref{subsec:classification}.

    \item\label{item:toeplitz} The operator $q^* p q$ can be identified with a block (or matrix-valued) Toeplitz operator with piecewise continuous symbol $\hat p$. Due to our assumption that the Hamiltonian is critical, $\hat p$ has at least one point of discontinuity on the circle. 
    We use arguments from the theory of block Toeplitz operators to show in section~\ref{subsec:toeplitz} that this suffices for $\spec(q^*pq)$ to contain a non-trivial interval. 
\end{enumerate}
We now proceed to explain steps \ref{item:classification} and \ref{item:toeplitz} of the argument.

\subsection{Classification Lemma}
\label{subsec:classification}
The following Lemma yields a sufficient condition to obtain type $\III_1$ for quasi-free states. It can be seen as a summary of the relevant results in \cite{powers1970free_states} and \cite{araki1968factors}.
    \begin{lemma}\label{lem:classification} Consider an operator $0\leq s\leq 1$ on $\h$ such that $[a,b]\subseteq \spec(s)$ for some $0\le a<b\le 1$. Then $\pi_{\omega_s}(\CAR(\h))''$ has type $\III_1$.
\end{lemma}
Note that the premise of the Lemma is fulfilled if the interval is contained in the essential spectrum $[a,b]\subseteq \essspec(s)\subseteq\spec(s)$.
To prove the Lemma, we make use of the following two results from \cite{powers1970free_states}. The first determines when two quasi-free states $\omega_s$ and $\omega_r$ are quasi-equivalent, in which case the respective factors are isomorphic. 
\begin{theorem}[{\cite[Thm.~5.1]{powers1970free_states}}]\label{thm:quasi-equivalence}
    Let $0\leq r,s\leq 1$ act on a Hilbert space $\h$. The quasi-free states $\omega_r$ and $\omega_s$ are quasi-equivalent if and only if $r^{1/2}-s^{1/2}$ and $(1-r)^{1/2}-(1-s)^{1/2}$ are Hilbert-Schmidt. 
\end{theorem}
The second result is a modification of a result by von Neumann showing that any self-adjoint operator on a separable Hilbert space may be approximated by one with a pure point spectrum up to an arbitrarily small error in the Hilbert-Schmidt norm. 
\begin{lemma}[{\cite[Lem.~4.3]{powers1970free_states}}]\label{lem:vonNeumann} Let $0\leq s\leq 1$ act on a Hilbert space $\h$ and $\eps>0$. Then there exists an operator $0\leq r\leq 1$ on $\h$ with pure point spectrum such that $r^{1/2}-s^{1/2}$ and $(1-r)^{1/2}-(1-s)^{1/2}$ are Hilbert-Schmidt with Hilbert-Schmidt norm less than $\eps$. Furthermore the eigenvalues of $r$ are dense in the spectrum of $s$.   
\end{lemma}

\begin{proof}[Proof of \cref{lem:classification}]
    By \cref{lem:vonNeumann}, there exists an operator $0\leq r\leq 1$ with pure point spectrum and eigenvalues that are dense in $\spec(s)$ and hence dense in $[a,b]$. By \cref{thm:quasi-equivalence}, the von Neumann algebras $\pi_{\omega_s}(\CAR(\h))''$ and $\pi_{\omega_r}(\CAR(\h))''$ have the same type. The factor $\pi_{\omega_r}(\CAR(\h))''$ is isomorphic to the ITPFI factor $\left(\otimes_j M_2(\CC), \otimes_j \mathrm{diag}(\mu_j,1-\mu_j)\right)$, where $\mu_j$ are the eigenvalues of $r$. 
    Since every $x\in[a,b]$ is an accumulation point of eigenvalues of $B$, the results of \cite{araki1968factors} show that $\pi_{\omega_s}(\CAR(\h))''$ has type $\III_1$ (type $S_\infty$ in their notation).
\end{proof}

\subsection{The essential spectrum of piecewise continuous block Toeplitz operators}
\label{subsec:toeplitz}
In this section, we show that the essential spectrum of the block Toeplitz operator $q^* p q$ contains a non-trivial interval. Together with \cref{lem:classification} this finishes the proof of \cref{thm:main-fermions}.

We will make use of the theory of piecewise continuous block Toeplitz operators.
To state these results we first have to introduce some notation. 
Let $\hat \Phi: S^1 \rightarrow M_b(\CC)$ be a left-continuous function with finitely many points of discontinuity.
To any such function, we can associate an operator $\Phi$ on $\h$ with $\hat \Phi$ as its symbol. Then $q^* \Phi q$ is a block Toeplitz operator. 
Define a function $\tilde\Phi: S^1 \times [0,1]\rightarrow M_n(\CC)$ by
\begin{align}
    \tilde\Phi(k,\mu) = \mu\hat\Phi(k) + (1-\mu)\hat\Phi(k+),\quad \hat\Phi(k+) := \lim_{\eps\rightarrow 0}\hat\Phi(k+ \eps).
\end{align}
Note that if $\hat \Phi$ is continuous at $k$, we have $\tilde\Phi(k,\mu)=\hat\Phi(k)$. If $\hat\Phi$ is discontinuous at $k$, $\mu \mapsto \tilde\Phi(k,\mu)$ linearly interpolates between $\hat\Phi(k)$ and $\hat\Phi(k+)$.
The following lemma is immediate from \cite[Thm.~3.1]{gohberg_classes_1993} since, by the definition, the essential spectrum of a bounded operator $A$ are those values $\lambda$ for which $A - \lambda 1$ is not Fredholm:
\begin{lemma}
    Let $\hat\Phi:S^1\rightarrow M_b(\CC)$ be piecewise continuous with finitely many points of discontinuity. Then $\lambda\in \essspec(q^* \Phi q)$ if and only if there exists a pair of points $(k,\mu) \in S^1\times [0,1]$ such that $\det(\tilde \Phi(k,\mu)-\lambda 1)=0$. In other words 
    \begin{align}
    \essspec(q^* \Phi q) = \bigcup_{k,\mu} \,\spec(\tilde \Phi(k,\mu)).
\end{align}
\end{lemma}
We will apply the theorem to $\hat p$, the symbol of the projection $p$. We can represent $\hat p$ by a piecewise continuous function $S^1\rightarrow \M_b(\CC)$ that is continuous from the left. Since $\hat p$ has at least one point of discontinuity, say $k_0\in[0,2\pi)$, we have that
\begin{align}
    \tilde p(k_0,\mu) = \mu \hat p(k_0) + (1-\mu) \hat p(k_0+),
\end{align}
where $\hat p(k_0)$ and $\hat p(k_0+)$ are distinct projection operators (otherwise $\hat p$ would be continuous at $k_0$). 
The following Lemma then shows that a non-trivial interval is contained in $\essspec(q^* \Phi q)$.

\begin{lemma}\label{lem:two-projection}
    Let $p\neq q$ be orthogonal projections on a finite-dimensional Hilbert space. Then
    there exist $a<b\in [0,1] $ with
    \begin{align}
        [a,b] \subseteq \bigcup_{\mu\in[0,1]} \spec\left(\mu p + (1-\mu)q\right).
    \end{align}
\end{lemma}
\begin{proof}
If $[p,q]=0$ and $p\neq q$ there exists a vector $\ket\psi$ such that (w.l.o.g.) $p\ket\psi = \ket\psi$ and $q\ket\psi = 0$. Then $\left(\mu p + (1-\mu)q\right)\ket\psi = \mu \ket\psi$ for $\mu\in[0,1]$ and we can choose $[a,b]=[0,1]$.
So let us assume $[p,q]\neq 0$.
It is well known from the theory of two subspaces \cite{halmos_two_1969} that $p$ and $q$ can be decomposed into orthogonal projections as $p=p_0+p_1$, $q=q_0+q_1$ with  $[p_0,q_0]=0$, $[p_1,q_1]\neq 0$ and $p_1,q_1$ are unitarily equivalent to projection operators of the form
\begin{align}
    p_1\cong \begin{pmatrix}1 & 0\\ 0& 0\end{pmatrix},\quad q_1 \cong \begin{pmatrix} c^2 & c(1-c^2)^{1/2} \\ c(1-c^2)^{1/2} & 1-c^2\end{pmatrix} 
\end{align}
on a Hilbert space $\mc K\oplus \mc K$ and where $0< c < 1$ (i.e. $0\le c\le 1$ and neither $0$ or nor $1$ are eigenvalues).
If $p_0\neq q_0$, the previous argument can be applied, and we get $[a,b]=[0,1]$. We can thus assume $p=p_1$ and $q=q_1$. Writing $c = \sum_j \chi_j \ketbra{\varphi_j}{\varphi_j}$ in a suitable orthonormal basis $\{\ket{\varphi_j}\}$ of $\mc K$, we find that
\begin{align}
    p\cong \bigoplus_{j} \begin{pmatrix} 1&0\\0&0\end{pmatrix},\quad q \cong \bigoplus_j \begin{pmatrix} \chi_j^2 & \chi_j (1-\chi_j^2)^{1/2}\\ \chi_j (1-\chi_j^2)^{1/2} & 1-\chi_j^2\end{pmatrix},\quad 0< \chi_k < 1.
\end{align}
The eigenvalues of $\mu p + (1-\mu)q$ are hence of the form $\lambda_\pm(\mu,\chi) = \frac{1}{2}\left(1 \pm \sqrt{4(\chi^2-1)(\mu-\mu^2) +1}\right)$ with $\chi\in \spec(c)$. These are continuous functions in $\mu$ with $\lambda_+(0,\chi) = \lambda_+(1,\chi) = 1,\lambda_-(0,\chi)=\lambda_-(1,\chi)=0$,   and $\lambda_\pm(1/2,\chi) = \frac{1}{2}(1\pm \chi)$. Hence $[(1+\norm{c})/2,1] \subset  \bigcup_{\mu\in[0,1]} \spec\left(\mu p + (1-\mu)q\right)$.
\end{proof}

\subsection{Embezzling families induced by embezzling states}
\label{sec:mbz_fam}

In this section, we show that embezzling families can be obtained from embezzling, infinite-volume, ground states of many-body systems with local Hamiltonians.
We derive our result from an abstract result relating embezzling states to embezzling families, which we discuss first.

Let us begin by formally recalling the definition of an embezzling family.
Since the fermionic observable algebras $\CAR_e(\h)$ are not pure matrix algebras but contain a center if $\dim\h<\oo$, we use an algebraic approach.
A bipartite system $(\M_A,\M_B,\H)$ is finite-dimensional if the Hilbert space $\H$ finite-dimensional.
A finite-dimensional bipartite system is necessarily of the form 
\begin{equation}\label{eq:direct_sum}
    \H = \bigoplus_{k=1}^N \CC^{n_k}\otimes \CC^{m_k}, \quad \M_A = \bigoplus_k M_{n_k}(\CC)\otimes 1, \quad \M_B= \bigoplus_k 1\ox M_{m_k}(\CC) 
\end{equation}
for some $N\in\NN$ and $n,m \in \NN^{N}$.

\begin{definition}\label{def:mbzfam-bipartite}
    Let $(\M_{A_n},\M_{B_n},\H_n)$ be a sequence of finite-dimensional bipartite systems.
    A sequence $(\ket{\Omega_n}_{A_nB_n})$ of bipartite state vectors $\ket{\Omega_n}_{AB}\in\H_n$ is a \emph{(bipartite) embezzling family} if for all $d\in\NN$ and $\eps>0$, there exists $n(\eps,d)$ such that for all $n\ge n(\eps,d)$ and all $\ket\Psi_{A'B'}\in\CC^d\ox\CC^d$ there exist unitaries $u_{A_nA'}\in \M_{A_n} \otimes M_d(\CC)\ox1$ and $u_{B_nB'}\in \M_{B_n}\ox1\ox M_d(\CC)$
    \begin{equation}
        \norm{\ket{\Omega_n}_{A_nB_n}\ox\ket0_{A'}\ket0_{B'}- u_{A_nA'}u_{B_nB'} \big(\ket{\Omega_n}_{A_nB_n}\ox\ket{\Psi}_{A'B'}\big)} < \eps.
    \end{equation}
\end{definition}

The reduced states of an embezzling family enjoy the following monopartite embezzling property:

\begin{definition}
\label{def:mbzfam}
    Let $(\M_n)$ be a sequence of finite-dimensional von Neumann algebras and $(\omega_n)$ be a sequence of states on $\M_n$.
    Then $(\omega_n)$ is a \emph{monopartite embezzling family} of states if for all $d\in\NN$ and all $\eps>0$ there exists $n(\eps,d)\in\NN$ such that for all $n\geq n(\eps,d)$ and every state $\psi$ on $M_d(\CC)$ there exist unitaries $u_{n}\in M_{d}(\M_{n}) = \M_{n}\ox M_{d}(\CC)$ such that
    \begin{align}\label{eq:vilsa}
        \norm{\omega_{n}\ox\psi-u_{n}\big(\omega_{n}\ox\bra0\placeholder\ket0\big)u_{n}^{*}} & < \eps.
    \end{align}
\end{definition}

\begin{proposition}\label{prop:mono-bi}
    Let $(\M_{A_n},\M_{B_n},\H)$ be a sequence of finite-dimensional bipartite systems. Let $\ket{\Omega_n}_{A_nB_n}\in\H_n$ be a sequence of bipartite state vectors and let $\omega_{A_n}$ be the sequence of induced states on $\M_{A_n}$.
    The following are equivalent:
    \begin{enumerate}[(a)]
        \item $(\ket{\Omega_n}_{A_nB_n})$ is a bipartite embezzling family,
        \item $(\omega_{A_n})$ is a monopartite embezzling family.
    \end{enumerate}
\end{proposition}

\begin{lemma}\label{lem:omivine}
    Let $(\M_A,\M_B,\H)$ be a bipartite system.
    Let $\ket\Omega_{AB}\in\H$ and $\ket\Psi_{A'B'}\in\CC^d\otimes\CC^d$ be bipartite states with reduced states $\omega_A$ and $\psi_A$ on $\M_A$ and $M_d(\CC)$, respectively.
    If there is a unitary $u_{AA'}\in \M_A\otimes M_d(\CC)$ such that $\norm{\omega_A\ox\psi_{A'}-u_{AA'}(\omega_A\ox\bra0\placeholder\ket0_{A'})u_{AA'}^*}\le \eps$, then there is a unitary $u_{BB'}\in \M_B\ox M_d(\CC)$ such that $\norm{\ket{\Omega}_{AB}\ox\ket{\Psi}_{A'B'}-u_{AA'} u_{BB'}(\ket{\Omega}_{AB}\ox\ket0_{A'}\ket0_{B'})}\le \eps^{1/2}$.
\end{lemma}
\begin{proof}
\emph{Step 1.} We begin by assuming that $(\M_A,\M_B,\H)$ is of the special form $(M_n(\CC)\ox1,1\ox M_m(\CC),\CC^n\ox\CC^m)$.
We make use of the following standard inequality, which immediately follows from the Powers-St\o rmer/Fuchs-van der Graaf inequalities:
Let $\ket{\Phi_1}, \ket{\Phi_2} \in \CC^n\ox \CC^m$ be two normalized vectors with reduced density matrices $\rho_1$ and $\rho_2$ on $\CC^n$, respectively. Then
\begin{align}\label{eq:graaf}
    \min_{u \in \mathrm U(m)} \norm{\ket{\Phi_1} - 1\otimes u \ket{\Phi_2}}^2 \leq \norm{\rho_1-\rho_2}_1.
\end{align}
The result then follows by setting $\ket{\Phi_1} = \ket{\Omega}_{AB}\ox \ket{\Psi}_{A'B'}, \ket{\Phi_2} = u_{AA'}\ket{\Omega}_{AB}\ox\ket 0_{A'} \ket 0_{B'}$.

\emph{Step 2.}
The general case follows from the first step since we can decompose a bipartite system $(\M_A,\M_B,\H)$ as a direct sum $\bigoplus (\M_{A_k},\M_{B_k},\H_k)$ of bipartite systems without a center, i.e., of the form considered in the first step, and since $\norm{\oplus_k \ket{\xi_k} - \oplus_k \ket{\eta_k}}^2 = \sum_k \norm{\ket{\xi_k} - \ket{\eta_k}}^2$ as well as $\norm{\oplus_k \rho_k - \oplus_k \sigma_k}_1 = \sum_k\norm{\rho_k-\sigma_k}_1$.
\end{proof}

\begin{proof}[Proof of \cref{prop:mono-bi}]
    It is clear that a bipartite embezzling family induces a monopartite embezzling family of reduced states.
    The converse follows from \cref{lem:omivine}.
\end{proof}

A general von Neumann algebra $\M$ is called \emph{hyperfinite} if it is the (weak) closure of an increasing family of finite type $\I$ von Neumann algebras, i.e., multi-matrix algebras of the form given in eq.~\eqref{eq:direct_sum}.
The following theorem shows that every embezzling state on a hyperfinite von Neumann algebra induces an embezzling family. 

\begin{theorem}
\label{thm:mbzfam}
    Let $\M$ be a hyperfinite von Neumann and let $\M_1\subset\M_2\subset...$ be a weakly dense increasing sequence of finite-dimensional von Neumann algebras. If $\omega$ is a monopartite embezzling state, the sequence $(\omega_n)$ of restricted states $\omega_n=\omega|_{\M_n}$ is a monopartite embezzling family.
\end{theorem}
To prove the theorem, we observe that the unitary group of $\M$ arises as the strong* closure of the unitary group of the increasing sequence of subalgebras:
\begin{lemma}
\label{lem:unitary_closure}
Consider $(\M_{n})$ and $\M$ as above, and denote $\M_{0}=\bigcup_{n}\M_{n}$. Then
\begin{align}
\label{eq:unitary_closure}
\U(\M) & = \overline{\U(\M_{0})}^{s^*}
\end{align}
where the closure on the right is taken in the strong* topology.
\end{lemma}

\begin{proof}[Proof (following Glimm {\cite[Lem.~3.1]{glimm1960uhf}}).]
In the norm topology, it is clear that $\overline{\U(\M_{0})}\subseteq\U(\overline{\M_{0}})$ because norm limits of unitaries are unitary. 
To see the converse, it suffices to show that every unitary $u\in \M_0$ arises as the norm limit of unitaries $u_n\in \M_n$.
Given $u\in\U(\M_{0})$, we find a sequence $(x_{n})$ with $x_{n}\in\M_{n}$ such that $\lim_{n}\norm{x_{n}-u}=0$. This implies that there exists an $N$ such that $x_{n}$ is invertible in $\M_{n}$ for all $n\geq N$ because the group of units of the Banach algebra $\M_{0}$ is open \cite[Prop.~I.1.7]{takesaki1}, and $x_{n}^{-1}$ can be obtained via the holomorphic functional calculus within the Banach algebra $\M_{n}$. Therefore, we can define the unitaries $u_{n}=x_{n}(x_{n}^{*}x_{n})^{-\frac{1}{2}}\in\M_{n}$ for $n\geq N$. It follows that $u_{n}\to u$ in norm because $(x_{n}^{*}x_{n})^{-\frac{1}{2}}\to(u^{*}u)^{-\frac{1}{2}}=1$ in norm.
To conclude the proof, we note that the unitary group of a C*-algebra $\A$ (acting on a Hilbert space) is strong* dense in the unitary group of the von Neumann algebra generated by $\A$ \cite[thm.~II.4.11]{takesaki1}.
\end{proof}

\begin{proof}[Proof of \cref{thm:mbzfam}]
Given $\eps>0$ and a state $\psi$ on $M_{m}(\CC)$, we find a unitary $u\in M_{m}(\M)$ such that
\begin{align}
\label{eq:mbzhalfeps}
    \norm{\omega\ox\psi-u(\omega\ox\bra0\placeholder\ket0)u^{*}} & < \eps.
\end{align}
By \cref{lem:unitary_closure}, the unitary $u$ can be approximated in the $\sigma$-strong topology (since the $\sigma$-strong and the strong topology coincide on the closed unit ball of $\M$ \cite[Lem.~II.2.5]{takesaki1})
by a sequence of unitaries $u_{j}\in\bigcup_{n}M_{m}(\M_{n})$.
It follows that\footnote{This is a consequence of the general fact that $\sigma$-strong convergence $x_n\to x$ implies $x_n\omega x_n^* \to x\omega x$ for all normal states $\omega$, which can be seen directly from the definition of the $\sigma$-strong topology \cite[Sec.~II.2]{takesaki1}.}
\begin{align*}
    \lim_{j}\norm{u_{j}(\omega\ox\bra0\placeholder\ket0)u_{j}^{*}-u(\omega\ox\bra0\placeholder\ket0)u^{*}} & = 0.
\end{align*}
In particular, we find an $j_{0}$ such that for all $j\geq j_{0}$, we have:
\begin{align}
\label{eq:uhalfeps}
    \norm{u_{j}(\omega\ox\bra0\placeholder\ket0)u_{j}^{*}-u(\omega\ox\bra0\placeholder\ket0)u^{*}} & < \eps.
\end{align}
Since $u_{j}\in\bigcup_{n}M_{m}(\M_{n})$ for all $j$, there is an $n_{0}$ such that $u_{j_{0}}\in M_{m}(\M_{n_{0}})$.
Combining \eqref{eq:mbzhalfeps} and \eqref{eq:uhalfeps}, and because $M_{m}(\M_{n_{0}})\subset M_{m}(\M_{n})$ for all $n\geq n_{0}$, we find:
\begin{align}
\label{eq:mbzfamcombined}
    \norm{\omega_{n}\ox\psi-u_{j_{0}}\big(\omega_{n}\ox\bra0\placeholder\ket0\big)u_{j_{0}}^{*}} & \leq \norm{\omega\ox\psi-u_{j_{0}}\big(\omega\ox\bra0\placeholder\ket0\big)u_{j_{0}}^{*}}\nonumber\\ \nonumber
    & \leq \norm{\omega\ox\psi-u\big(\omega\ox\bra0\placeholder\ket0\big)u^{*}} \\ \nonumber
    & \hspace{0.5cm} + \norm{u_{j_{0}}\big(\omega\ox\bra0\placeholder\ket0\big)u_{j_{0}}^{*}-u(\omega\ox\bra0\placeholder\ket0)u^{*}}\\ 
    & < \eps.
\end{align}
To conclude the argument, we observe that for every $\eps>0$, we can choose a finite collection of states $\{\psi_{i}\}_{i\in I}$ on $M_{m}(\CC)$, $|I|<\infty$, such that every state $\psi$ on $M_{m}(\CC)$ has norm distance $<\eps$ to this collection (due to the norm compactness of the states space of $M_{m}(\CC)$). We refer to such a collection $\{\psi_{i}\}_{i\in I}$ as a (finite) $\eps$-cover. Thus, we find for each $i\in I$, unitaries $u_{j_{0}}(i)$ and $n_{0}(i)$ such that \eqref{eq:mbzfamcombined} holds for any state $\psi$ on $M_{m}(\CC)$ if $n\geq\max_{i\in I}n_{0}(i)$.
This proves the claim.
\end{proof}

The assumptions of \cref{thm:mbzfam} can be relaxed by allowing for some additional freedom in the relation between the embezzling family $\omega^{(n)}$ and the embezzling state $\omega$ as stated in the following corollary.
\begin{corollary}
\label{cor:mbzfam}
With the assumption of \cref{thm:mbzfam}, let $\omega^{(n)}$ be a sequence of states on $\M_n$ that converges to $\omega$ in the sense that $\lim_{n}\omega^{(n)}(a_k)=\omega(a_{k})$ for all $k$ and $a_{k}\in\M_{k}$. Then, there is a subsequence that is an embezzling family.
\end{corollary}
\begin{proof}
Set $\omega_k = \omega|_{\M_k}$ and $\omega^{(n)}_k = \omega^{(n)}|_{\M_k}$ for $n\geq k$.
Since each $\M_{k}$ is finite-dimensional, we know that $\lim_{n}\norm{\omega^{(n)}_{k}-\omega_{k}} = 0$ as weak* and norm convergence are equivalent.
Using the same notation as in the proof of \cref{thm:mbzfam}, we can estimate for a state $\psi$ on $M_{m}(\CC)$ and for all $k\geq k_{0}$ using \eqref{eq:mbzfamcombined}:
\begin{align} \nonumber
\norm{\omega^{(n)}_{k}\ox\psi-u_{j_{0}}\big(\omega^{(n)}_{k}\ox\bra0\placeholder\ket0\big)u_{j_{0}}^{*}} & \leq 2\norm{\omega^{(n)}_{k}-\omega_{k}}+\norm{\omega_{k}\ox\psi-u_{j_{0}}\big(\omega_{k}\ox\bra0\placeholder\ket0\big)u_{j_{0}}^{*}} \\ 
    & < 2\norm{\omega^{(n)}_{k}-\omega_{k}} + \tfrac{\eps}{2}.\label{eq:mbzfam_approx}
\end{align}
Thus, if we choose a large enough $n=n_{k}$ such that $\norm{\omega^{(n)}_{k}-\omega_{k}}<\tfrac{\eps}{4}$, we can construct the required subsequence by choosing $\omega^{(k)} = \omega^{(n_{k})}$ on $\M_{n_{k}}$. As in the proof of \cref{thm:mbzfam}, we obtain a uniform estimate for all states $\psi$ on $M_{m}(\CC)$ by choosing $k$ and, therefore, $n=n_{k}$ so large that eq.~\eqref{eq:mbzfam_approx} holds with respect to an $\eps$-cover.
\end{proof}

Now, let us discuss how \cref{thm:mbzfam} applies to many-body systems. For background material on the rigorous description of infinite many-body systems, we refer to \cite{bratteli1987oa1,bratteli1997oa2}.
A general many-body system on a lattice $\Lambda$ (e.g., $\Lambda = \ZZ$) is described by its dynamics on the quasi-local algebra $\A$, which is a C*-algebra generated by the subalgebras $\A(X)$ describing observables localized in a finite region $X\subset \Lambda$ such that i) $\A(X) \subseteq \A(Y)$ if $X\subseteq Y$ and ii) $[\A(X),\A(Y)]=0$ if $X\cap Y=\emptyset$.
For example, for fermionic systems $\A(X)$ would be $\CAR_e(p_X\h)$, where $p_X$ is a projection onto the single-particle subspace associated to the region $X$ and where subspaces associated with disjoint regions are orthogonal.
For spin systems $\A(X) = \bigotimes_{x\in X} M_d(\CC)$.
In the following, we only assume that the $\A(X)$ are finite-dimensional for finite regions $X$. 
The dynamics of the infinite system induced by a Hamiltonian $H$ is described by a one-parameter group of automorphisms $\tau_t$ on $\A$.
A ground state can be characterized abstractly as a state $\omega$ on $\A$ that is a KMS state of temperature zero with respect to $\tau_t$ \cite[Def. 5.3.18]{bratteli1997oa2}.

If we only consider a finite sublattice $\Lambda_n$, for example $\Lambda_n = \{-n,\ldots,n\}\subset\ZZ$, and if $H$ is sufficiently local, we can define Hamiltonians $H_n$ acting non-trivially on $\A(\Lambda_n)$ (and trivially on operators supported outside of $\Lambda_n$) by only considering the local terms of $H$ acting within $\Lambda_n$. 
In our case, the Hamiltonian $H_n$ is necessarily bounded and induces a strongly continuous one-parameter group of automorphisms $\tau^{(n)}_t$ on $\A$. By construction, $\tau^{(n)}_t$ leaves invariant all operators with support outside of $\Lambda_n$.  

We assume that the dynamics $\tau^{(n)}_t$ converges strongly to $\tau_t$: For each $A\in \A$ we have\begin{align}
    \lim_{n\rightarrow \infty} \norm{\tau^{(n)}_t(A) - \tau_t(A)} = 0.
\end{align}
In fact, typically, $\tau_t$ is derived in this way, for example, using Lieb-Robinson bounds. We refer to \cite{bratteli1997oa2,nachtergaele_quasi-locality_2019,van_luijk_convergence_2024} and references therein for further details. 

Let us now fix a disjoint bipartition of the lattice $\Lambda$ as $\Lambda = \Lambda_A\cup \Lambda_B$. Given a ground state $\omega$ of $H$ and  its GNS representation $(\pi_\omega,\H,\ket\Omega_{AB})$, we obtain von Neumann algebras
\begin{align}
    \M_A := \pi_\omega(\A(\Lambda_A))'',\quad \M_B := \pi_\omega(\A(\Lambda_B))''.
\end{align}
In the following, we assume that the GNS representation $\pi_\omega$ is faithful. This means that there are no irrelevant physical degrees of freedom in the quasi-local algebra $\A$. In particular, this is true if $\A$ is simple, which is true for spin systems and fermionic systems.
Assuming Haag duality, $\M_A = \M_B'$, we obtain a bipartite system $(\M_A,\M_B,\H)$.  
Similarly, for each $\Lambda_n$ we can consider a pure ground state $\omega^{(n)}$ of the Hamiltonian $H_n$ restricted to $\A(\Lambda_n)$ and GNS representation $(\pi_n, \H_n,\ket{\Omega_n}_{AB})$.  We then obtain (finite-dimensional) algebras
\begin{align}
    \M^{(n)}_{A_n} := \pi_{n}(\A(\Lambda_n\cap \Lambda_A)), \quad \M^{(n)}_{B_n} := \pi_{n}(\A(\Lambda_n\cap \Lambda_B)).
\end{align}
If Haag duality holds this yields a bipartite system $(\M^{(n)}_{A_n},\M^{(n)}_{B_n},\H_n)$. 
For fermionic systems, Haag duality holds by \cref{prop:duality} because twisted duality holds for any even state if $\h$ is finite-dimensional. 
For spin systems, Haag duality also holds on this level because $\A(\Lambda_n) = \A(\Lambda_n\cap \Lambda_A)\ox \A(\Lambda_n\cap \Lambda_B)$ with factors $\A(\Lambda_n\cap \Lambda_i)$.

Then, we have the following corollary, which implies \cref{thm:main-families}:

\begin{corollary}\label{cor:gs-families} Suppose under the above conditions that the ground state $\omega$ is unique and represented by an embezzling state $\ket\Omega_{AB}$ on the bipartite system $(\M_A,\M_B,\H)$.
Then the sequence of bipartite systems $(\M^{(n)}_{A_n},\M^{(n)}_{B_n},\H_n)$ and states $\ket{\Omega_n}_{AB}\in\H_n$ contains a bipartite embezzling family as subsequence.
\end{corollary}
\begin{proof}
In the following, we use the notation $A_n:= \Lambda_n \cap \Lambda_A$.
    The automorphisms $\tau^{(n)}_t$ on $\A$ converge to $\tau_t$ strongly by assumption.
    We can extend the chosen ground states $\omega^{(n)}$ of $H_n$ on $\A(\Lambda_n)$ to states $\tilde \omega^{(n)}$ on $\A$ that are also ground states of $\tau^{(n)}_t$ \cite[Prop.~2.3.24]{bratteli1987oa1}. 
    By \cite[Prop.~5.3.25]{bratteli1997oa2}, every weak* cluster point of the sequence $\tilde\omega^{(n)}$ is an infinite-volume ground state. Since, by assumption, the infinite-volume dynamics has a unique ground state $\omega$, it follows that $\tilde\omega\up n$ weak*-converges to $\omega$. 
    Consequently, the restrictions of $\omega^{(n)}$ to $\A(A_k)$, which coincide with the restrictions $\tilde\omega^{(n)}|_{\A(A_k)}$ of $\tilde\omega^{(n)}$ for $n\ge k$, weak*-converge to the restriction $\omega|_{\A(A_k)}$ of $\omega$ to $\A(A_k)$.
    Since the relevant observable algebras are finite-dimensional, we have, in fact:
    \begin{align}
    \lim_n \,\norm{\omega^{(n)}|_{\A(A_k)} - \omega|_{\A(A_k)}} = 0,\quad  k\in \NN. 
    \end{align}
    Since $\pi_\omega$ is faithful by assumption, we can consider the states $\omega^{(n)}|_{A_n}$ as states on the (von Neumann) algebras $\M_{A_n} := \pi_\omega(\A(A_n))$, which converge  to $\omega$ in the sense of \cref{cor:mbzfam}. Consequently, the sequence $\omega^{(n)}_{A_n}$ contains a subsequence that is a monopartite embezzling family. 
     Considering $\omega^{(n)}_{A_n}$ again as states on $\M^{(n)}_{A_n}$, we find that the sequence $(\M^{(n)}_n,\omega^{(n)}_{A_n})$ has a monopartite embezzling family as subsequence.
     By \cref{prop:mono-bi} a subsequence of the $\ket{\Omega_n}$, then also provides a bipartite embezzling family.
\end{proof}
\begin{remark}
\label{rem:uniqueness}
    Note that in the preceding corollary, the assumption that $\omega$ is unique is only used to argue that the local ground states converge towards an embezzling state. If this can be ensured by different means, for example, because all ground states are embezzling, then the assumption is not required.
\end{remark}

\printbibliography

\end{document}

%% file: common_preamble.tex
\usepackage[utf8]{inputenc}
\usepackage[T1]{fontenc}
\usepackage[english]{babel}
\usepackage{amsmath,amssymb,amsthm}
\usepackage{graphicx} 
\usepackage[margin=1.in]{geometry}
\usepackage[shortlabels]{enumitem}
\usepackage{todonotes}
\usepackage[bibstyle=numeric, backend=biber, sorting=none, citestyle=numeric-comp, giveninits,url=false,maxbibnames=5]{biblatex} 
\usepackage{csquotes}\MakeOuterQuote"
\usepackage{soul}\setstcolor{red}
\usepackage[colorlinks,citecolor=blue,linkcolor=blue]{hyperref}
\usepackage[capitalize]{cleveref}
\usepackage{comment}

\newtheorem{theorem}{Theorem}
\newtheorem{lemma}[theorem]{Lemma}
\newtheorem{corollary}[theorem]{Corollary}
\newtheorem{proposition}[theorem]{Proposition}

\newtheorem{definition}[theorem]{Definition}

\theoremstyle{definition}

\newtheorem*{problem*}{Problem}
\newtheorem{assumption}[theorem]{Assumption}
\newtheorem{example}[theorem]{Example}
\newtheorem{remark}[theorem]{Remark}
\newtheorem*{warning*}{Warning}

\newcommand{\F}{{\mc F}}

\newcommand{\braket}[2]{\langle #1|#2\rangle}
\newcommand{\ketbra}[2]{|#1\rangle\langle#2|}
\newcommand{\ket}[1]{|#1\rangle}

\newcommand{\bra}[1]{\langle#1|}

\newcommand{\norm}[1]{\lVert #1\rVert}
\newcommand{\oo}{\infty}
\newcommand{\ox}{\otimes}
\newcommand{\mc}{\mathcal}
\newcommand{\eps}{\varepsilon}
\newcommand{\III}{{\mathrm{III}}}
\newcommand{\II}{{\mathrm{II}}}
\newcommand{\I}{{\mathrm{I}}}
\newcommand{\e}{\mathrm{e}}

\newcommand{\up}[1]{^{(#1)}}

\DeclareMathOperator{\tr}{Tr}

\renewcommand{\tilde}{\widetilde}
\renewcommand{\hat}{\widehat}

\newcommand{\hide}[1]{}

\def\A{{\mc A}}

\def\B{{\mc B}}
\def\CC{{\mathbb C}}

\def\H{{\mc H}}

\def\M{{\mc M}}
\def\N{{\mc N}}

\def\U{{\mc U}}

\def\RR{{\mathbb R}}

\def\NN{{\mathbb N}}

\def\ZZ{{\mathbb Z}}

\DeclareMathOperator{\id}{id}

\DeclareMathOperator{\spec}{Sp}
\DeclareMathOperator{\essspec}{\spec_{\mathrm{ess}}}

\newcommand{\CAR}{\textup{CAR}}

\newcommand*{\1}{\text{\usefont{U}{bbold}{m}{n}1}}
\newcommand{\placeholder}[0]{{\,\cdot\,}}

